\documentclass[letter,12pt]{article}
\usepackage{amsmath}
\usepackage{enumerate}
\usepackage{amssymb}
\usepackage{geometry}
\usepackage{setspace}
\usepackage{graphicx}
\usepackage{natbib}
\usepackage{lmodern}
\usepackage{soul}
\usepackage{comment}
\usepackage{ifthen}
\usepackage{xr}
\usepackage{framed}
\newtheorem{theorem}{Theorem}

\newtheorem{definition}{Definition}
\newtheorem{lemma}{Lemma}
\newtheorem{proposition}{Proposition}
\newtheorem{remark}{Remark}

\newenvironment{proof}[1][Proof]{\textbf{#1.} }{\ \rule{0.5em}{0.5em}}
\DeclareMathSizes{12}{11}{7.7}{5.5}
\usepackage[inline]{enumitem}
\usepackage{comment}

\DeclareMathOperator*{\esssup}{ess\,sup}

\graphicspath{ {WD-Graphs/} }

\begin{document}

\title{\textbf{ On the Time-Inconsistent Deterministic Linear-Quadratic Control}
}

\thispagestyle{empty}
\title{On the Time-Inconsistent Deterministic Linear-Quadratic Control  }
\author{Hongyan Cai \and Danhong Chen \and Yunfei Peng \and Wei Wei
\thanks{Cai: School of Mathematics and Statistics, Guizhou University, Guiyang 550025, China E-mail: hycai@gzu.edu.cn. Chen: School of Mathematics and Statistics, Guizhou University, Guiyang 550025, China. E-mail: gs.dhchen19@gzu.edu.cn. Peng: School of Mathematics and Statistics, Guizhou University, Guiyang 550025, China. Email: yfpeng@gzu.edu.cn. Wei: Department of Actuarial Mathematics and Statistics, Heriot-Watt University, Edinburgh, Scotland EH14 4AS, UK. E-mail: wei.wei@hw.ac.uk. Peng gratefully acknowledges financial support through the National Natural Science Foundation of China under grant No. 12061021. Wei gratefully acknowledges financial support through a start-up grant at Heriot-Watt University.} }
%
%
\thispagestyle{empty}
\date{\today}
\maketitle
\thispagestyle{empty}

\vspace*{-0.4cm}
\begin{abstract}
A fundamental theory of deterministic linear-quadratic (LQ) control is the equivalent relationship between control problems, two-point boundary value problems and Riccati equations. In this paper, we extend the equivalence to a general time-inconsistent deterministic LQ problem, where the inconsistency arises from non-exponential discount functions. By studying the solvability of the Riccati equation, we show the existence and uniqueness of the linear equilibrium for the time-inconsistent LQ problem.
\end{abstract}


\newpage
\setcounter{page}{1}
\section{Introduction}\label{SecIntroduction}
Experimental and empirical studies suggest that decision makers often behave in a time-inconsistent manner, with individuals acting impatiently in the current moment whilst planning to act
patiently in the future\footnote{In economics, \cite{str1955} first observed that non-constant time preference rates result in time-inconsistent decisions.}. In order to incorporate this evidence on time-inconsistency into mathematical modelling, behavioural scientists and economists suppose that the discount functions in the control models have non-exponential forms\footnote{See \cite{ekelaz2006}, \cite{har1995}, \cite{kar2007}, \cite{lai1997}, \cite{loewenstein1992anomalies}, \cite{bleichrodt2009non} and \cite{ebert2020weighted} for various discount functions. }, and in consequence, the optimal control becomes inconsistent in time. As is standard in the literature on decision making, time-inconsistent problems are often considered within the self-control and intra-personal game theoretic framework and the corresponding equilibria are taken as solutions to such problems.\footnote{See, for example, \cite{phepol1968}; \cite{lai1997}; \cite{odorab1999}; \cite{krusmi2003} and \cite{lutmar2003}.}\par
This paper studies a general deterministic time-inconsistent linear-quadratic (LQ) control problem, where the inconsistency arises from non-exponential discount functions. We demonstrate two main results. First, we show an equivalence relationship between control problems, two-point boundary value problems and Riccati equations. Second, we establish the existence and uniqueness of the linear equilibrium for the time-inconsistent LQ problem. 
We wish these results could shed some light on the study of general time-inconsistent control problems, and justify, at least for some extent, the definition of equilibria over a set of feedback controls. \par 
A fundamental theory of classical (time-consistent) LQ optimal control is the equivalent relationship between control problems, two-point boundary value problems and Riccati equations. 
In the time consistent setting, the equivalence between the three problems is a natural result of the spike variation and the linear-quadratic structure inherent to the problem\footnote{See, for example, Chapter 6 of \cite{yong1999stochastic}.}. 
However, a different pictrue emerges when the time consistency is lost. Due to the time inconsistency stemming from the continuously changing time preferences, a time parameter is involved in the Riccati equation that characterises the LQ problem. As a consequence, the Riccati equation turns out to be a flow of dynamics, which switches continuously over time\footnote{A similar phenomenon has been observed in the literature studying stochastic time-inconsistent LQ problems. See, for example, \cite{hu2012time}. }. This feature places an obstacle in the way of using some powerful techniques in the ODE theory and thus constituting the major difficulty in the study of the time-inconsistent LQ problem.
In the present paper, we extend the equivalence to the time-inconsistent LQ setting. Particularly, using a different method from the spike variation, we obtain a Riccati equation, which admits a symmetric structure and solves the time-inconsistent LQ problem. In contrast to the Riccati equation derived from the spike variation, the Riccati equation obtained in this paper does not involve the time parameter and hence characterising a single dynamics. Using the Banach fixed point theory and the extension technique in the ODE theory, we establish the unique solvability of the Riccati equation. This result, together with the equivalence, yields the existence and the uniqueness of the linear equilibrium for the time-inconsistent LQ problem.\par 
There are studies in the literature on time-inconsistent LQ control within the intra-personal game theoretic framework. \cite{basak2010dynamic} and  \cite{bjomurzho2014} study dynamic mean-variance asset allocation problems which can be formulated within the time-inconsistent LQ framework, where the inconsistency comes from the quadratic term of the expected state and (or) the non-constant risk aversion. In their works, the definition of an equilibrium is based on the notion of feedback control, which is formally proposed by \cite{ekelaz2006} and \cite{bjomur2009}. Particularly, this type of equilibria fit in the results in most of literature on behavioural economics, in which the equilibria are derived based on the recursive method\footnote{See the discussion in \cite{bjomur2009}.}. The search of this type of equilibria is usually conducted through a verification theorem and a complicated Bellman system, of which the uniqueness of the solution remains unknown. As a result, the uniqueness of such an equilibrium is usually absent.
A different type of definition of an equilibrium, which is based on the notion of open loop control, is introduced by \cite{hu2012time}. Moreover, \cite{hu2017time} prove the uniqueness of the open loop equilibrium for a time-inconsistent LQ problem in the one dimensional case\footnote{The time inconsistency of the LQ problem discussed in  \cite{hu2012time} and \cite{hu2017time} also arises from the quadratic term of the expected state.}. The third method to define (or construct) an equilibrium is discrete approximation, which is given by \cite{yong2012deterministic}, where the equilibrium of a continuous time-inconsistent LQ control is defined as the limit of the solutions to a sequence of discrete time problems. In the same line, \cite{dou2020time} construct a feedback equilibrium for a time-inconsistent for a time-inconsistent LQ problem, where the underlying process takes values in a Hilbert space. Another related stream of research is in the literature combining time-inconsistent LQ control with other topics, such as mean field games. The readers could be referred to \cite{bensoussan2013linear}, \cite{yong2017linear} and the reference therein. Finally, it is worth noting that both non-existence and non-uniqueness results for general time-inconsistent control problems have been reported in literature. For a time-inconsistent binary control problem with non-exponential discounting, \cite{tan2018failure} find that an equilibrium may not exists, while the time-consistent counterpart admits a unique optimal solution. For a behavioural portfolio management problem, \cite{ekepir2008} show that multiple linear equilibria can be found over a feedback control set, even though the value functions are obtained from the same ansatz. \par
The remainder  of the paper is organized as follows. Section \ref{Sec:ProblemSetting} introduces the formulations of the time-inconsistent LQ problem and the definition of equilibria. In particular, our definition is consistent with the definition based on the feedback control, proposed by \cite{ekelaz2006} and \cite{bjomur2009}. Section \ref{Sec:MainResult} demonstrates the equivalence between the time-inconsistent LQ problem, the two boundary value ODE system and the Riccati equation. Section \ref{Sec:Riccati} shows the unique solvability of the Riccati equation and thus providing the existence and uniqueness of the linear equilibrium for the equivalent time-inconsistent LQ problem. In section \ref{Sec:Conclusion}, concluding remarks are given.
\section{Problem setting}\label{Sec:ProblemSetting}
Let $T>0$ be the end of a finite time horizon. Throughout the paper, we will use the following notations.
\begin{align*}
L^{p}\left([0,T];\mathbb{R}^{l\times k}\right) = \left\{f:[0,T]\to\mathbb{R}^{l\times k}\vert \int^T_0 \vert f_{ij}(t)\vert^p dt<\infty, 1\le i\le l, 1\le j\le k. \right\}
\end{align*}
\begin{align*}
L^{\infty}\left([0,T];\mathbb{R}^{l\times k}\right) =\left \{f:[0,T]\to\mathbb{R}^{l\times k}\vert \esssup_{t\in[0,T]}| f_{ij}(t)|<\infty, 1\le i\le l, 1\le j\le k. \right\}
\end{align*}
\begin{align*}
 C([0,T]^m; \mathbb{R}^{l\times k})=\left\{f:[0,T]^m\to\mathbb{R}^{l\times k}\vert \text{$f$ is continuous.}\right\}
\end{align*}
\begin{align*}
 C^{1}([0,T]^m; \mathbb{R}^{l\times k})=\left\{f:[0,T]^m\to\mathbb{R}^{l\times k}\vert \text{$f$ and its first order derivative are continuous.}\right\}
\end{align*}
For a real matrix-valued function $O(t)=(o_{ij}(t))\in \mathbb{R}^{l\times k}$ for all $t\in [0,T]^{m} (m=1,2)$, we introduce the following norms,
\begin{eqnarray*}
\left\{\begin{array}{ll}
\|O(t)\|=\max\limits_{1\leq i\leq l}\sum\limits_{j=1}^k|o_{ij}(t)|,\;\;
\|O\|_{L^1}=\max\limits_{1\leq i\leq l}\sum\limits_{j=1}^k\|o_{ij}\|_{L^1([0,T]^{m})}, \;\;\|O\|_{L^{\infty}}=\max\limits_{1\leq i\leq l}\sum\limits_{j=1}^k\|o_{ij}\|_{L^{\infty}([0,T]^{m})}, \\
\|O\|_{C}=\max\limits_{1\leq i\leq l}\sum\limits_{j=1}^k\|o_{ij}\|_{C([0,T]^{m})},\;\;\|O\|_{C^1}=\max\limits_{1\leq i\leq l}\sum\limits_{j=1}^k\left(\|o_{ij}\|_{C([0,T]^{m})}+\|Do_{ij}\|_{C([0,T]^{m})}\right).
\end{array}\right.
\end{eqnarray*}
We suppose the following assumptions hold throughout this paper.
\begin{itemize}
\item[(H1)] $A\in L^{1}\left((0,T);\mathbb{R}^{n\times n}\right)$, $B\in L^{\infty}\left((0,T);\mathbb{R}^{n\times m}\right)$.
\item[(H2)] $M\in C^{1}([0,T]\times[0,T]; \mathbb{R}^{m\times m})$ is a positive  definite symmetric matrix-valued function.
\item[(H3)] $Q\in C^{1}([0,T]\times[0,T];\mathbb{R}^{n\times n}) $ and $G\in C^{1}([0,T];\mathbb{R}^{n\times n}) $ are   positive semi-definite symmetric matrix-valued functions.
\item[(H4)]  $S\in C^{1}([0,T]\times[0,T]; \mathbb{R}^{m\times n})$.
\end{itemize}
 For any $(t,x)\in[0,T)\times \mathbb{R}^{n}$,  the underlying dynamics is governed by the following controlled linear ordinary differential equation (LODE, for short)\footnote{Without any specification, any vector in this paper is a column vector.}
 \begin{equation} \label{1.1}
 \left\{
 \begin{array}{ll}
 \dot{X}(s)=A(s)X(s)+B(s)u(s),& s\in(t,T], \\
 X(t)=x,
 \end{array}
 \right.
 \end{equation}
where the function  $u\in \mathcal{U}[0,T]  \equiv L^2 \left([0,T];\mathbb{R}^m\right)$ is the control and $X$ is the state process valued in $\mathbb{R}^n.$ It follows standard ODE theory that LODE (\ref{1.1}) has a unique solution $X(\cdot) \equiv X^{u}_{t,x}(\cdot)$.
At any time $t$ with the system state $X_t = x$, the cost functional is given by
 \begin{eqnarray} \label{1.2}
 J(t,x;u)&=&\int_{t}^{T}\bigg[\left\langle Q(t,s)X^{u}_{t,x}(s),X^{u}_{t,x}(s)\right\rangle+2\left\langle
S(t,s)X^{u}_{t,x}(s),u(s)\right\rangle\bigg]ds\nonumber\\
&&+\int_{t}^{T}\left\langle M(t,s)u(s),u(s)\right\rangle ds+\left\langle G(t)X^{u}_{t,x}(T),X^{u}_{t,x}(T)\right\rangle.
\end{eqnarray}
As discussed earlier the non-exponential discount functions $Q, S, M, G$ in the cost functional (\ref{1.2}) render
the underlying LQ problem generally time-inconsistent. In this paper we consider
a sophisticated agent who is aware of the time-inconsistency but unable to
control her future actions. In this case, she seeks to find the so-called equilibrium strategies
within the intra-personal game theoretic framework, in which the individual is represented by
different players at different dates.
We now give the precise definition of an equilibrium control $\bar{u}$, which essentially entails
a solution to a game in which no self at any time (or, equivalently in the current setting, at any
state) is willing to deviate from $\bar{u}$.\par 
\begin{definition}\label{DefEquilibrium}
Let $\bar{u}:[0,T]\times \mathbb{R}^{n}\longrightarrow \mathbb{R}^{m}$ be a measurable mapping and satisfies $\bar{u}(\cdot,X(\cdot))\in \mathcal{U}[0,T]$. Define
\begin{equation} \label{1.4}
u^{\varepsilon,v}(s,x)=\left\{\begin{array}{ll}
v,&s\in [t,t+\varepsilon],\\
\bar{u}(s,x),&s\in (t+\varepsilon,T].
\end{array}
\right.
\end{equation}
The control $\bar{u}$ is an equilibrium control if
\begin{equation} \label{1.3}
 \liminf_{\varepsilon\searrow 0} \frac{J\left(t,x;u^{\varepsilon,v}\right)-J(t,x;\bar{u})}{\varepsilon}\geq 0,\;\;\forall (t,x)\in [0, T )\times \mathbb{R}^{n}, v\in \mathbb{R}^{m}.
 \end{equation}
 Furthermore, $\bar{u}$  is called a linear equilibrium control if there is a function $\tilde{u}:[0,T] \longrightarrow\mathbb{R}^{m\times n}$ such that
 \begin{equation*} 
\bar{u}(t,x)=\tilde{u}(t)x \mbox{ for any } (t,x)\in[0,T]\times \mathbb{R}^{n}.
 \end{equation*}
\end{definition}

\section{The equivalence}\label{Sec:MainResult}
In this section, we focus on the equivalent relationship between control problems, two-point boundary value problems and Riccati equations. \par
First, we introduce the following two-point boundary value problem
\begin{equation} \label{1.6}
 	\left\{	
 	\begin{array}{ll}
 	\dot{\bar{X}}(s)=\left[A(s)-B(s)M^{-1}(s,s)S(s,s)\right]\bar{X}(s)-B(s)M^{-1}(s,s)B^{\top}(s)\varphi(s),& s\in(t,T], \\
\dot{\varphi}(s)=-\left[A(s)-B(s)M^{-1}(s,s)S(s,s)\right]^{\top}\varphi(s)\\
\quad-\left[\hat{Q}(s,s)-S^{\top} (s,s) M^{-1}(s,s)S(s,s)\right]\bar{X}(s),& s\in[t,T),\\
\bar{ X}(t)=x,\varphi(T)=G(T) \bar{X}(T),
 	\end{array}
 	\right.
 	\end{equation}
where
\begin{eqnarray}\label{1.7}
	\left\langle\hat{Q}(s,s)\bar{ X}(s),\bar{ X}(s)\right\rangle
&=&\left\langle Q(s,s)\bar{ X}(s),\bar{ X}(s)\right\rangle-\left\langle \dot{G}(s)\bar{ X}(T),\bar{ X}(T)\right\rangle-\int_{s}^{T}\left\langle Q_{s}(s,\tau)\bar{ X}(\tau),\bar{ X}(\tau)\right\rangle d\tau\nonumber\\
&&-\int_{s}^{T}\bigg\langle M_{s}(s,\tau)M^{-1}(\tau,\tau)\left(B^{\top}(\tau)\varphi(\tau)+S(\tau,\tau)\bar{ X}(\tau)\right)-2S_{s}(s,\tau)\bar{ X}(\tau),\\
&&M^{-1}(\tau,\tau)\left(B^{\top}(\tau)\varphi(\tau)+S(\tau,\tau)\bar{ X}(\tau)\right)\bigg\rangle d\tau\nonumber
\end{eqnarray}
for any $s\in[t,T]$ and for any $y\in \mathbb{R}^n$.

Second, the Riccati equation is given by
 	\begin{equation} \label{1.8}
 	\left\{
 	\begin{array}{ll}
\dot{P}(t)+A^{\top}(t)P(t)+P(t)A(t)+\bar{Q}(t,t)\\
-\left[P(t)B(t)+S^{\top}(t,t)\right]M^{-1}(t,t)\left[B^{\top}(t)P(t)+S(t,t)\right]=0,& 0\leq t< T,\\
P(T)=G(T).
 	\end{array}
 	\right.
 	\end{equation}
Here
\begin{eqnarray}\label{1.9}
\bar{Q}(t,t)&=&Q(t,t)-\Phi^{\top}(T,t)\dot{G}(t)\Phi(T,t)-\int_{t}^{T}\Phi^{\top}(s,t)Q_t(t,s)\Phi(s,t)ds\\
&&-\int_{t}^{T}\Phi^{\top}(s,t)\bigg[\Upsilon^{\top}(s)M_t(t,s) \Upsilon(s)-\Upsilon^{\top}(s)S_t(t,s)-S_t^{\top}(t,s)\Upsilon(s)\bigg]\Phi(s,t)ds,\nonumber
\end{eqnarray}
where
\begin{align}\label{Phi1}
\left\{\begin{array}{ll}\Upsilon(s)=M^{-1}(s,s)\left(B^{\top}(s)P(s)+S(s,s)\right), ~~\forall s\in[0,T],\\
\Phi(t,s)=\exp\left(\int_{s}^{t}\left(A(\tau)-B(\tau)\Upsilon(\tau)\right)d\tau\right),~~\forall s,t\in[0,T].
\end{array}\right.
\end{align}
\begin{remark}
There is an abuse of terminology. Compared to the conventional Riccati equations in time-consistent LQ models, the Riccati equation (\ref{1.8}) has an integral term which contains the solution to the equation itself (See (\ref{1.9}) and (\ref{Phi1}).). This feature highlights the difference between time-inconsistent LQ control and its time-consistent counterpart. For the reader's convenience, instead of introducing a new terminology, we still call equation (\ref{1.8}) (equilibrium) Riccati equation.
\end{remark}
As is standard literature on classical time consistent LQ control and Riccati equations, we define the solution to the Riccati differential equation (\ref{1.8}) in $C\left([0,T];\mathbb{R}^{n\times n}\right)$ as follows.
\begin{definition}\label{DefRiccati}
A symmetric matrix-valued function $P\in C\left([0,T];\mathbb{R}^{n\times n}\right)$ is a solution to the equilibrium Riccati differential equation (\ref{1.8}) if $P$ satisfies the following integral equation,
\begin{align*}
P(t) &= G(T) +\int^T_t (A^{\top}(s)P(s)+P(s)A(s)+\bar{Q}(s,s)\nonumber\\&
-\left[P(s)B(s)+S^{\top}(s,s)\right]M^{-1}(s,s)\left[B^{\top}(t)P(s)+S(s,s)\right] )ds
\end{align*}
\end{definition}
The result given as follows establishes the equivalence between the time-inconsistent LQ problem (\ref{1.3}), the two-point boundary value problem (\ref{1.6}) and the Riccati equation (\ref{1.8}).
\begin{theorem}\label{peng-theorem1.1}The following statements are equivalent.
\begin{itemize}
\item[(i)] A linear equilibrium control defined by Definition \ref{DefEquilibrium} exists.
\item[(ii)] For $\forall (t,x)\in [0, T )\times \mathbb{R}^{n},$ the  two-point equilibrium boundary value problem (\ref{1.6})	
 	admits a solution in $C([t,T];\mathbb{R}^n)\times C([t,T];\mathbb{R}^n)$.
 	\item [(iii)] The  equilibrium Riccati equation (\ref{1.8})
 	admits a symmetric solution $P(\cdot)\in C([0,T];\mathbb{R}^{n\times n})$.
\end{itemize}
\end{theorem}
The proof of Theorem \ref{peng-theorem1.1} is lengthy. For the ease of exposition, we divide the theorem into three propositions as follows.
\begin{proposition}\label{peng-proposition2.1}
Suppose that the time-inconsistent LQ problem in Definition \ref{DefEquilibrium} admits a linear equilibrium control, then the  two-point equilibrium boundary value problem (\ref{1.6})	
admits a solution in $C([t,T];\mathbb{R}^n)\times C([t,T];\mathbb{R}^n)$ for any $t\in[0,T)$.
\end{proposition}
\begin{proof}
We prove the proposition with two steps. In Step $1$, we derive a representation of the equilibrium, given the existence of the equilibrium. In Step $2$, using the representation obtained in step $1$, we construct a solution to the two-point equilibrium boundary value problem.\par
\textit{Step $1$}. Let $\bar{u}$ denote the linear equilibrium control defined by Definition \ref{DefEquilibrium}, then there exists a function $\tilde{u}\in L^{2}\left((0,T);\mathbb{R}^{n\times m}\right)$ such that
	\begin{equation}\label{LinearControl} 
	\bar{u}(t,x)=\tilde{u}(t)x \mbox{ for any } (t,x)\in[0,T)\times \mathbb{R}^{n}.
	\end{equation}
Consider \begin{equation*} \label{2.p3}
	Y^{\varepsilon,t,v}(s)=\frac{X^{u^{\varepsilon,v}}_{t,x}(s)-X^{\bar{u}}_{t,x}(s)}{\varepsilon},~~s\in [t,T].
	\end{equation*}
Then it follows from the linear structure of the equilibrium control (\ref{LinearControl}), the definition of the perturbation control (\ref{1.4}) and the LODE (\ref{1.1}) that	
\[\left\{\begin{array}{ll}
\dot{Y}^{\varepsilon,t,v}(s)=\left\{\begin{array}{ll}
A(s)	Y^{\varepsilon,t,v}(s)+\frac{1}{\varepsilon}
\left[B(s)v-B(s)\tilde{u}(s)\bar{X}_{t,x}(s)\right],&s\in[t,t+\varepsilon],\\
(A(s)+B(s)\tilde{u}(s))	Y^{\varepsilon,t,v}(s),&s\in (t+\varepsilon,T],
\end{array}\right.\\
Y^{\varepsilon,t,v}(t)=0.
\end{array}
\right.
\]		
The above ODE problem admits a unique solution $Y^{\varepsilon,t,v}\in C\left([0,T];\mathbb{R}^n\right)$ given by
\begin{eqnarray} \label{2.5}
Y^{\varepsilon,t,v}(s)
	&=&\left\{\begin{array}{ll}
\frac{1}{\varepsilon }\int_{t}^{s}\Phi_{A}(s,\nu)B(\nu)[v-\tilde{u}(\nu)\tilde{\Phi}(\nu,t)x]d\nu,&s\in[t,t+\varepsilon],\\
	\frac{1}{\varepsilon }\tilde{\Phi}(s,t+\varepsilon)\int_{t}^{t+\varepsilon}\Phi_{A}(t+\varepsilon,\nu)B(\nu)[v-\tilde{u}(\nu)\tilde{\Phi}(\nu,t)x]d\nu,&s\in (t+\varepsilon,T],
	\end{array}\right.
	\end{eqnarray}
where	
\begin{equation} \label{2.4}
	\Phi_{A}(s,t)=\exp\left(\int_{t}^{s}A(\nu)d\nu\right),\tilde{\Phi} (s,t)=\exp\left(\int_{t}^{s}(A(\nu)+B(\nu)\tilde{u}(\nu))d\nu\right) ,\forall t,s\in [0,T].
	\end{equation}
As $X^{u^{\varepsilon,v}}_{t,x}(s)-X^{\bar{u}}_{t,x}(s) = \varepsilon Y^{\varepsilon,t,v}(s)$, then (\ref{2.5}) yields that
	\begin{equation*} \label{2.6}
\lim\limits_{\varepsilon\searrow0}\left\|X^{u^{\varepsilon,v}}_{t,x}-X^{\bar{u}}_{t,x}\right\|_{C([t,T];\mathbb{R}^n)}=0.
\end{equation*}
For ease of exposition, we introduce the following functions.
	\begin{equation} \label{2.3}
\tilde{Q}(t,s)=
Q(t,s)+\tilde{u}^{\top}(s)S(t,s)+S^{\top}(t,s)\tilde{u}(s)+\tilde{u}^{\top}(s)M(t,s)\tilde{u}(s),~~s\in [t,T].
\end{equation}
and
\begin{eqnarray} \label{2.10}
\tilde{P} (\tau)=\tilde{\Phi}^{\top}(T,\tau)G(t)\tilde{\Phi} (T,t)+ \int_{\tau}^{T} \tilde{\Phi}^{\top}(s,\tau)\tilde{Q}(t,s)\tilde{\Phi} (s,t)ds,\tau\in [t,T].
\end{eqnarray}
We are now calculating $\lim\limits_{\varepsilon\searrow 0} \frac{J\left(t,x;u^{\varepsilon,v}\right)-J(t,x;\bar{u})}{\varepsilon}$.
\begin{align*}
\lim_{\varepsilon\searrow 0} \frac{J\left(t,x;u^{\varepsilon,v}\right)-J(t,x;\bar{u})}{\varepsilon}
	&=\left\langle Q(t,t)x,x\right\rangle+2\left\langle
	S(t,t)x,v\right\rangle+\left\langle M(t,s)v,v\right\rangle\\\nonumber &-\lim_{\varepsilon\searrow 0}\frac{1}{\varepsilon}\int_{t}^{t+\varepsilon}\left\langle \tilde{Q}(t,s)X^{\bar{u}}_{t,x}(s),X^{\bar{u}}_{t,x}(s)\right\rangle ds\nonumber\\
	&+\lim_{\varepsilon\searrow 0}\int_{t+\varepsilon}^{T}\left\langle \tilde{Q}(t,s)\left(X^{\bar{u}}_{t,x}(s)+X^{u^{\varepsilon,v}}_{t,x}(s)\right),Y^{\varepsilon,t,v}(s)\right\rangle ds\\\nonumber &+\lim_{\varepsilon\searrow 0}\left\langle G(t)\left(X^{\bar{u}}_{t,x}(T)+X^{u^{\varepsilon,v}}_{t,x}(T)\right),Y^{\varepsilon,t,v}(T)\right\rangle.\nonumber
\end{align*}
Plug the representations of $X^{\bar{u}}_{t,x}(s)$, $X^{u^{\varepsilon,v}}_{t,x}(s)$ into the above equation, we then have
\begin{align*}
\lim_{\varepsilon\searrow 0} \frac{J\left(t,x;u^{\varepsilon,v}\right)-J(t,x;\bar{u})}{\varepsilon}&=\left\langle Q(t,t)x,x\right\rangle+2\left\langle
	S(t,t)x,v\right\rangle+\left\langle M(t,s)v,v\right\rangle\\\nonumber &-\lim_{\varepsilon\searrow 0}\frac{1}{\varepsilon}\int_{t}^{t+\varepsilon}\left\langle \tilde{\Phi}^{\top} (s,t)\tilde{Q}(t,s)\tilde{\Phi} (s,t)x,x\right\rangle ds\nonumber\\
	&+2\lim_{\varepsilon\searrow 0}\int_{t+\varepsilon}^{T}\left\langle \tilde{Q}(t,s)\tilde{\Phi} (s,t)x,Y^{\varepsilon,t,v}(s)\right\rangle ds+2\lim_{\varepsilon\searrow 0}\left\langle G(t)\tilde{\Phi} (T,t)x,Y^{\varepsilon,t,v}(T)\right\rangle.\nonumber\\
\end{align*}
Plug the representation of $Y^{\varepsilon,t,v}$ into the above equation, we then have
	\begin{eqnarray} \label{2.7}
	&&\lim_{\varepsilon\searrow 0} \frac{J\left(t,x;u^{\varepsilon,v}\right)-J(t,x;\bar{u})}{\varepsilon}\nonumber\\
&=&\left\langle Q(t,t)x,x\right\rangle+2\left\langle
S(t,t)x,v\right\rangle+\left\langle M(t,s)v,v\right\rangle-\lim_{\varepsilon\searrow 0}\frac{1}{\varepsilon}\int_{t}^{t+\varepsilon}\left\langle \tilde{\Phi}^{\top} (s,t)\tilde{Q}(t,s)\tilde{\Phi} (s,t)x,x\right\rangle ds\nonumber\\
&&+2\left\langle \lim_{\varepsilon\searrow 0}\frac{1}{\varepsilon}\int_{t}^{t+\varepsilon} B^{\top}(\nu)\Phi^{\top}_{A}(t+\varepsilon,\nu)\tilde{P} (t+\varepsilon)xd\nu,v\right\rangle \\
&&-2\left\langle \lim_{\varepsilon\searrow 0}\frac{1}{\varepsilon}\int_{t}^{t+\varepsilon}\tilde{\Phi}^{\top}(\nu,t)\tilde{u}^{\top}(\nu)]B^{\top}(\nu)\Phi^{\top}_{A}(t+\varepsilon,\nu)\tilde{P} (t+\varepsilon)xd\nu,x\right\rangle.\nonumber
	\end{eqnarray}
	Define \begin{align*}
	\tilde{J}(t,x,v)\equiv\lim_{\varepsilon\searrow 0} \frac{J\left(t,x;u^{\varepsilon,v}\right)-J(t,x;\bar{u})}{\varepsilon}.
	\end{align*}
For any fixed $ (t,x)\in [0,T)\times\mathbb{R}^{n}$, it follows from (\ref{2.7}) that $\tilde{J}(t,x,v)$ is strictly convex in $v$.\par
Moreover, the definition of an equilibrium yields that $\tilde{J}(t,x,v)\ge 0$, which, together with (\ref{2.7}), implies that $\tilde{J}(t,x,v)$ has a unique minimum point  $\tilde{v}$ given by
	\begin{eqnarray*} \label{2.9}
	\tilde{v}=-M^{-1}(t,t)\left[S(t,t)+\lim_{\varepsilon\searrow 0}\frac{1}{\varepsilon}\int_{t}^{t+\varepsilon} B^{\top}(\nu)\Phi^{\top}_{A}(t+\varepsilon,\nu)\tilde{P} (t+\varepsilon)d\nu\right]x.
	\end{eqnarray*}
	It follows from (\ref{2.4}) and (\ref{2.10}) that
	\begin{eqnarray*}\label{2.wp1}
	\lim_{\varepsilon\searrow 0}\frac{1}{\varepsilon}\int_{t}^{t+\varepsilon} B^{\top}(\nu)\Phi^{\top}_{A}(t+\varepsilon,\nu)\tilde{P} (t+\varepsilon)d\nu=B^{\top}(t)\tilde{P}(t), \mbox{ a.e. }t\in[0,T],
\end{eqnarray*}
and thus
\begin{align*}
\tilde{v} = -M^{-1}(t,t)\left(B^{\top}(t)\tilde{P}(t)+S(t,t)\right)x.
\end{align*}
Then the uniqueness of the minimum point $\tilde{v}$ yields that
	\begin{eqnarray} \label{2.11}
	\bar{u}(t,x)=\tilde{v}=-M^{-1}(t,t)\left(B^{\top}(t)\tilde{P}(t)+S(t,t)\right)x\mbox{ for any }(t,x)\in [0,T]\times\mathbb{R}^{n}.
	\end{eqnarray}
\textit{Step $2$.} Define
\begin{eqnarray} \label{2.14}
	\varphi(s)=\tilde{P}(s)X^{\bar{u}}_{t,x}(s)\mbox{ for any }s\in [t,T].
	\end{eqnarray}
	We will verify that $(X^{\bar{u}}_{t,x},\varphi)$ is a solution to the  two-point equilibrium boundary value problem (\ref{1.6}).\par
Following Assumptions (H1)-(H4),  and the representation of $\tilde{P}$ (\ref{2.10}), it is to see that $\tilde{P}$ is absolutely continuous.
Take the first order derivative on $\tilde{P}$, then we have
	\begin{eqnarray} \label{2.13}
	\dot{\tilde{P}}(t)&=&-(A(t)+B(t)\tilde{u}(t))^{\top}\tilde{P}(t)-\tilde{P}(t)(A(t)+B(t)\tilde{u}(t))-\tilde{Q}(t,t)\\
	&&+\tilde{\Phi}^{\top}(T,t)\dot{G}(t)\tilde{\Phi}(T,t)+\int_{t}^{T}\tilde{\Phi}^{\top}(s,t)\tilde{Q}_t(t,s)\tilde{\Phi}(s,t) ds, \forall t\in[0,T].\nonumber
	\end{eqnarray}
	Plug  (\ref{2.11}) into LODE (\ref{1.1}), we then have
	\begin{eqnarray} \label{2.15}
	\dot{X}^{\bar{u}}_{t,x}(s)
	=\left[A(s)-B(s)M^{-1}(s,s)S(s,s)\right]X^{\bar{u}}_{t,x}(s)-B(s)M^{-1}(s,s)B^{\top}(s)\varphi(s).
	\end{eqnarray}
	Take the first order derivative on $\varphi$, then (\ref{2.11}), (\ref{2.14}),  (\ref{2.13}) and (\ref{2.15}) yield that
	\begin{eqnarray*} \label{2.16}
	\dot{\varphi}(s)
	&=&-\left[A(s)-B(s)M^{-1}(s,s)\left(B^{\top}(s)\tilde{P}(s)+S(s,s)\right)\right]^{\top}\varphi(s)\nonumber\\
	&&+\tilde{P}(s)\left[A(s)-B(s)M^{-1}(s,s)\left(B^{\top}(s)\tilde{P}(s)+S(s,s)\right)\right]X^{\bar{u}}_{t,x}(s)-\tilde{Q}(s,s)X^{\bar{u}}_{t,x}(s)\nonumber\\
	&&+\tilde{\Phi}^{\top}(T,s)\dot{G}(s)\tilde{\Phi}(T,s)X^{\bar{u}}_{t,x}(s)+\int_{s}^{T}\tilde{\Phi}^{\top}(\tau,s)\tilde{Q}_{s}(s,\tau)\tilde{\Phi}(\tau,s) X^{\bar{u}}_{t,x}(s)d\tau\nonumber\\
	&&-\tilde{P}(s)\left[A(s)-B(s)M^{-1}(s,s)S(s,s)\right]X^{\bar{u}}_{t,x}(s)+\tilde{P}(s)B(s)M^{-1}(s,s)B^{\top}(s)\varphi(s)\\
	&=&-\left[A(s)-B(s)M^{-1}(s,s)S(s,s)\right]^{\top}\varphi(s)+\left[B(s)M^{-1}(s,s)B^{\top}(s)\tilde{P}(s)\right]^{\top}\varphi(s)\nonumber\\
	&&-\tilde{Q}(s,s)X^{\bar{u}}_{t,x}(s)+\tilde{\Phi}^{\top}(T,s)\dot{G}(s)\tilde{\Phi}(T,s)X^{\bar{u}}_{t,x}(s)+\int_{s}^{T}\tilde{\Phi}^{\top}(\tau,s)\tilde{Q}_{s}(s,\tau)\tilde{\Phi}(\tau,s) X^{\bar{u}}_{t,x}(s)d\tau\nonumber\\
	&=&-\left[A(s)-B(s)M^{-1}(s,s)S(s,s)\right]^{\top}\varphi(s)-\left(\hat{Q}(s,s)-S^{\top}(s,s)M^{-1}(s,s)S(s,s)\right)X^{\bar{u}}_{t,x}(s),\nonumber
	\end{eqnarray*}
	where
	\begin{eqnarray} \label{2.17}
	\hat{Q}(s,s)&=&\tilde{Q}(s,s)-\tilde{P}^{\top}(s)B(s)M^{-1}(s,s)B^{\top}(s)\tilde{P}(s)+S^{\top}(s,s)M^{-1}(s,s)S(s,s)\nonumber\\
	&&-\tilde{\Phi}^{\top}(T,s)\dot{G}(s)\tilde{\Phi}(T,s)-\int_{s}^{T}\tilde{\Phi}^{\top}(\tau,s)\tilde{Q}_{s}(s,\tau)\tilde{\Phi}(\tau,s) d\tau.
	\end{eqnarray}	
 (\ref{2.3}) and (\ref{2.17}) lead to
	\begin{eqnarray} \label{2.18}
	\hat{Q}(s,s)&=&Q(s,s)-\tilde{\Phi}^{\top}(T,s)\dot{G}(s)\tilde{\Phi}(T,s)-\int_{s}^{T}\tilde{\Phi}^{\top}(\tau,s)\tilde{Q}_{s}(s,\tau)\tilde{\Phi}(\tau,s) d\tau,\nonumber\\
	\tilde{Q}_{s}(s,\tau)&=&Q_{s}(s,\tau)-S^{\top}_{s}(s,\tau)M^{-1}(\tau,\tau)\left(B^{\top}(\tau)\tilde{P}(\tau)+S(\tau,\tau)\right)\\
	&&-\left(B^{\top}(\tau)\tilde{P}(\tau)+S(\tau,\tau)\right)^{\top}M^{-1}(\tau,\tau)S_{s}(s,\tau)\nonumber\\
	&&+\left(B^{\top}(\tau)\tilde{P}(\tau)+S(\tau,\tau)\right)^{\top}M^{-1}(\tau,\tau)M_{s}(s,\tau)M^{-1}(\tau,\tau)\left(B^{\top}(\tau)\tilde{P}(\tau)+S(\tau,\tau)\right).\nonumber
	\end{eqnarray}
Moreover, it follows from (\ref{2.14}) and (\ref{2.18}) that
	\begin{eqnarray*}
		&&\left\langle\hat{Q}(s,s)X^{\bar{u}}_{t,x}(s),X^{\bar{u}}_{t,x}(s)\right\rangle\\
		&=&\left\langle Q(s,s)X^{\bar{u}}_{t,x}(s),X^{\bar{u}}_{t,x}(s)\right\rangle-\left\langle \dot{G}(s)X^{\bar{u}}_{t,x}(T),X^{\bar{u}}_{t,x}(T)\right\rangle-\int_{s}^{T}\left\langle \tilde{Q}_{s}(s,\tau)X^{\bar{u}}_{t,x}(\tau),X^{\bar{u}}_{t,x}(\tau)\right\rangle d\tau\nonumber\\
		&=&\left\langle Q(s,s)X^{\bar{u}}_{t,x}(s),X^{\bar{u}}_{t,x}(s)\right\rangle-\left\langle \dot{G}(s)X^{\bar{u}}_{t,x}(T),X^{\bar{u}}_{t,x}(T)\right\rangle\\
		&&-\int_{s}^{T}\left\langle Q_{s}(s,\tau)X^{\bar{u}}_{t,x}(\tau),X^{\bar{u}}_{t,x}(\tau)\right\rangle d\tau\\
		&&-\int_{s}^{T}\bigg\langle M_{s}(s,\tau)M^{-1}(\tau,\tau)\left(B^{\top}(\tau)\varphi(\tau)+S(\tau,\tau)X^{\bar{u}}_{t,x}(\tau)\right)-2S_{s}(s,\tau)X^{\bar{u}}_{t,x}(\tau),\\
		&&M^{-1}(\tau,\tau)\left(B^{\top}(\tau)\varphi(\tau)+S(\tau,\tau)X^{\bar{u}}_{t,x}(\tau)\right)\bigg\rangle d\tau
	\end{eqnarray*}
	Finally, combining (\ref{2.14}), (\ref{2.15}) and (\ref{2.16}) together, we have that $(X^{\bar{u}}_{t,x},\varphi)\in C\left([t,T];\mathbb{R}^{n}\right)\times C\left([t,T];\mathbb{R}^{n}\right)$ solves the two-point boundary value problem (\ref{1.6}) for any $t\in[0,T)$. This completes the proof.
\end{proof}
\begin{proposition}\label{peng-proposition2.2}
	 Suppose that the  two-point equilibrium boundary value problem (\ref{1.6}) has a solution $(\varphi(\cdot), \bar{X}(\cdot))\in C\left([t,T];\mathbb{R}^n\right)\times C\left([t,T];\mathbb{R}^n\right)$ for any $(t,x)\in [0,T)\times \mathbb{R}^n$, then the equilibrium Riccati equation (\ref{1.8}) has a symmetric solution $P(\cdot)\in C\left([0,T];\mathbb{R}^{n\times n}\right)$.
\end{proposition}
\begin{proof}
For any fixed $x\in \mathbb{R}^n$ with $x \neq 0$, denote the solution to  the two point boundary problem (\ref{1.6}) (with initial state $(0,x)$) by $(\varphi,\bar{X})\in C([0,T];\mathbb{R}^n)\times C([0,T];\mathbb{R}^n)$.  We define
	\begin{eqnarray}\label{2.22}
\varphi(t)=P(t)\bar{X}(t)\mbox{ for all }t\in [0,T].
\end{eqnarray}
In order to make sure the existence of such $P(t)$ in (\ref{2.22}), we need to show $\bar{X}(t)\neq0, \forall t\in[0,T]$. In fact, plug  (\ref{2.22}) into  (\ref{1.6}), we then have that
$\bar{X}$ satisfies the following differential equation
\[
\left\{
\begin{array}{ll}
\dot{\bar{X}}(t)=\left[A(t)-B(t)M^{-1}(t,t)\left(B^{\top}(t)P(t)+S(t,t)\right)\right]\bar{X}(t),& t\in(0,T], \\
\bar{ X}(0)=x.
\end{array}
\right.
\]
Solving the above ordinary differential equation, we then have
\[\bar{X}(t)=\Phi(t,0)x\mbox{ for all } t\in [0,T],\]
where $\Phi(t,0)$ is given by (\ref{Phi1}).\par
Hence,  $\bar{X}(t)\neq0, \forall t\in[0,T]$ follows from $x\neq 0.$\par
Next, we derive the differential equation that $P(t)$ satisfies. As $\varphi$ and $\bar{X}$ are continuous and weakly differentiable, we have that $P\in W^{1,1}\left((0,T);\mathbb{R}^{n\times n}\right)$. Furthermore, taking the first order derivative on the both sides of  (\ref{2.22}), we then have
	\[\dot{\varphi}(s)=\dot{P}(s)\bar{X}(s)+P(s)\dot{\bar{X}}(s)\mbox{ for all }s\in[0,T].\]
	Moreover, it is from (\ref{1.6}), (\ref{1.7}) and (\ref{2.22})  to see
	\begin{eqnarray}\label{2.23}
	\left\{\begin{array}{ll}\dot{P
	}(s)+P(s)A(s)+A^{\top}(s)P(s)+\hat{Q}(s,s)\\
-\left[P(s)B(s)+S^{\top}(s,s)\right]M^{-1}(s,s)\left[B^{\top}(s)P(s)+S(s,s)\right]=0,&s\in [0,T),\\
P(T)=G(T).
\end{array}\right.
	\end{eqnarray}
Next, we verify that $\bar{Q}(t,t) = \hat{Q}(t,t)$.  It follows from  (\ref{1.7}) that
\begin{eqnarray*}
&&\left\langle\hat{Q}(s,s)\Phi(s,0)x,\Phi(s,0)x\right\rangle\\
&=&\left\langle Q(s,s)\Phi(s,0)x,\Phi(s,0)x\right\rangle-\left\langle \dot{G}(s)\Phi(T,0)x,\Phi(T,0)x\right\rangle\\
&&-\int_{s}^{T}\left\langle Q_s(s,\tau)\Phi(\tau,0)x,\Phi(\tau,0)x\right\rangle d\tau\nonumber\\
&&-\int_{s}^{T}\bigg\langle  M_s(s,\tau)M^{-1}(\tau,\tau)\left(B^{\top}(\tau)P(\tau)+S(\tau,\tau)\right)\Phi(\tau,0)x\\
&&-2S_s(s,\tau)\Phi(\tau,0)x,
M^{-1}(\tau,\tau)\left(B^{\top}(\tau)P(\tau)+S(\tau,\tau)\right)\Phi(\tau,0)x\bigg\rangle d\tau\nonumber\\
&=&\left\langle Q(s,s)\Phi(s,0)x,\Phi(s,0)x\right\rangle-\left\langle \Phi^{\top}(T,s)\dot{G}(s)\Phi(T,s)\Phi(s,0)x,\Phi(s,0)x\right\rangle\\
&&-\int_{s}^{T}\left\langle\Phi^{\top}(\tau,s) Q_s(s,\tau)\Phi(\tau,s)\Phi(s,0)x,\Phi(s,0)x\right\rangle d\tau\nonumber\\
&&-\int_{s}^{T}\bigg\langle \Phi^{\top}(\tau,s)\Upsilon^{\top}(\tau) \bigg[ M_s(s,\tau)\Upsilon(\tau)\Phi(\tau,s)-2S_s(s,\tau)\Phi(\tau,s)\bigg]\Phi(s,0)x,
\Phi(s,0)x\bigg\rangle d\tau.\nonumber
\end{eqnarray*}
Then the representation $\bar{Q}(t,t)$ (see (\ref{1.9})) yields that
\begin{align*}
\left\langle\hat{Q}(s,s)\Phi(s,0)x,\Phi(s,0)x\right\rangle = \left\langle\bar{Q}(s,s)\Phi(s,0)x,\Phi(s,0)x\right\rangle ,
\end{align*}
which suggests that  $\bar{Q}(t,t) = \hat{Q}(t,t)$.\par
Therefore, comparing (\ref{2.23}) and (\ref{1.8}), we have $P$ solves the Riccati equation (\ref{1.8}).\par
It now suffices to prove the symmetry of $P$. It is easy to see from (\ref{1.9}) that  $\bar{Q}^{\top} (t,t)=\bar{Q}(t,t)$ for all $t\in [0,T]$.  Moreover, the Riccati equation  (\ref{1.8}) yields that
	\begin{eqnarray}\label{p2.23}
	\left\{\begin{array}{ll}
	\dot{P}^{\top}(s)+P^{\top}(s)A(s)+A^{\top}(s)P^{\top}(s)+\bar{Q}(s,s)\\
	-\left[B^{\top}(s)P(s)+S(s,s)\right]^{\top}M^{-1}(s,s)\left[B^{\top}(s)P^{\top}(s)+S(s,s)\right]=0,&s\in [0,T),\\
	P^{\top}(T)=G(T).
	\end{array}\right.
	\end{eqnarray}
	Define
\begin{eqnarray}\label{0}
\left\{\begin{array}{ll}
\mathbb{P}(s)=P^{\top}(s)-P(s),&s\in [0,T],\\
\mathbb{A}(s)=A(s)-B(s)M^{-1}(s,s)\left[B^{\top}(s)P^{\top}(s)+S(s,s)\right],&s\in [0,T].
\end{array}\right.
\end{eqnarray}
One can see from (\ref{2.23}) and (\ref{p2.23}) that
	\begin{eqnarray}\label{p0.23}
	&&\frac{d}{ds}\left(P^{\top}(s)-P(s)\right)+\left(P^{\top}(s)-P(s)\right)A(s)+A^{\top}(s)\left(P^{\top}(s)-P(s)\right)\nonumber\\
	&=&\left(P^{\top}(s)-P(s)\right)B(s)M^{-1}(s,s)\left[B^{\top}(s)P^{\top}(s)+S(s,s)\right]\\
	&&+\left[B^{\top}(s)P^{\top}(s)+S(s,s)\right]^{\top}M^{-1}(s,s)B^{\top}(s)\left(P^{\top}(s)-P(s)\right).\nonumber
	\end{eqnarray}
Furthermore, it follows from (\ref{2.23}), (\ref{p2.23}), (\ref{0}) and (\ref{p0.23}) that
\begin{eqnarray*}
\left\{\begin{array}{ll}
\dot{\mathbb{P}}(s)+\mathbb{P}(s)\mathbb{A}(s)+\mathbb{A}^{\top}(s)\mathbb{P}(s)=0,&s\in [0,T),\\
\mathbb{P}(T)=0.
\end{array}\right.
\end{eqnarray*}
It is easy to see that (\ref{p0.23}) admits a unique solution $\mathbb{P}(s)=0$ for all $s\in [0,T]$, which is equivalent to $P^{\top}(s)=P(s)$ for all $s\in [0,T]$.\par  This completes the proof.
\end{proof}

\begin{proposition}\label{peng-proposition2.3} If the equilibrium Riccati equation (\ref{1.8}) has a  symmetric solution $P(\cdot)\in C\left([0,T];\mathbb{R}^{n\times n}\right)$, then the problem  has an equilibrium control $\bar{u}$ given by
	\begin{eqnarray}\label{2.24}
	\bar{u}(t,x)=-M^{-1}(t,t)\left(B^{\top}P(t)+S(t,t)\right)x,~~\forall(t,x)\in[0,T]\times\mathbb{R}^{n}.
	\end{eqnarray}
\end{proposition}
\begin{proof}
	Let $P$ solves the equilibrim Riccati equation (\ref{1.8}), then
	\begin{eqnarray}\label{2.26}
	&&\dot{P}(s)+P(s)\left[ A(s)-B(s)M^{-1}(s,s)\left(B^{\top}(s)P(s)+S(s,s)\right)\right]\nonumber\\
	&&+\left[A(s)-B(s)M^{-1}(s,s)\left(B^{\top}(s)P(s)+S(s,s)\right)\right]^{\top}P(s)+\bar{Q}(s,s)\\
	&=&S^{\top}(s,s)   M^{-1}(s,s)S(s,s) -P(s)B(s)M^{-1}(s,s)B^{\top}(s)P(s), \mbox{ for any }s\in[0,T).\nonumber
	\end{eqnarray}
	Plug the feedback control $\bar{u}$ given by (\ref{2.24}) into the LODE (\ref{1.1}), we then have
	\[
	\left\{
	\begin{array}{ll}
	\dot{\bar{X}}(s)=A(s)\bar{X}(s)-B(s)M^{-1}(s,s)\left(B^{\top}(s)P(s)+S(s,s)\right)\bar{X}(s),& s\in(t,T], \\
	\bar{ X}(t)=x,
	\end{array}
	\right.
	\]
	which admits the unique solution $\bar{X}$ given by
	\begin{equation*} \label{2.27}
	\bar{X}(s)=X_{t,x}^{\bar{u}}(s)=\Phi(s,t)x\mbox{, } \forall(t,x)\in [0,T)\times\mathbb{R}^{n}, s\in[t,T].
	\end{equation*}
	It follows from (\ref{2.26}) and the representation of $\bar{X}$ that
	\begin{eqnarray}\label{w1}
		&&\frac{d}{ds}\langle P(s)\bar{X}(s),\bar{X}(s)\rangle+\left\langle P(s)B(s)M^{-1}(s,s)B^{\top}(s)P(s)\bar{X}(s),\bar{X}(s)\rangle\right\rangle\\
		&=&\left\langle S^{\top}(s,s)   M^{-1}(s,s)S(s,s) \bar{X}(s),\bar{X}(s)\rangle\right\rangle-\langle\bar{Q}(s,s)\bar{X}(s),\bar{X}(s)\rangle, \mbox{ for any }s\in[t,T).\nonumber
	\end{eqnarray}
Integrate (\ref{w1}) from $t$ to $T$ and plug the representations (\ref{1.8}),  (\ref{1.9}) and (\ref{2.24}) into the integral, we then have
	\begin{eqnarray*}
		\langle P(t)x,x\rangle&=&\langle G(T)\bar{X}(T),\bar{X}(T)\rangle+\int_{t}^{T}\left\langle M^{-1}(s,s)B^{\top}(s)P(s)\bar{X}(s),B^{\top}(s)P(s)\bar{X}(s)\right\rangle ds\\
		&&-\int_{t}^{T}\left\langle M^{-1}(s,s)S(s,s) \bar{X}(s),S(s,s)\bar{X}(s)\rangle\right\rangle ds+\int_{t}^{T}\langle\bar{Q}(s,s)\bar{X}(s),\bar{X}(s)\rangle ds
		\\
		&=&\int_{t}^{T}\langle Q(s,s)\bar{X}(s),\bar{X}(s)\rangle ds-\int_{t}^{T}\left\langle M^{-1}(s,s)S(s,s) \bar{X}(s),S(s,s)\bar{X}(s)\right\rangle ds\\
		&&+\langle G(T)\bar{X}(T),\bar{X}(T)\rangle+\int_{t}^{T}\left\langle M^{-1}(s,s)B^{\top}(s)P(s)\bar{X}(s),B^{\top}(s)P(s)\bar{X}(s)\right\rangle ds\\
		&&-\int_{t}^{T}\left\langle \dot{G}(s)\bar{X}(T),\bar{X}(T)\right\rangle ds-\int_{t}^{T}\int_{s}^{T}\left\langle Q_s(s,\tau)\bar{X}(\tau),\bar{X}(\tau)\right\rangle d\tau ds\\
		&&-\int_{t}^{T}\int_{s}^{T}\left\langle \left[M_s(s,\tau) \Upsilon(\tau)-2S_s(s,\tau)\right]\bar{X}(\tau),\Upsilon(\tau) \bar{X}(\tau)\right\rangle d\tau ds\\
		&=&\int_{t}^{T}\langle Q(t,s)\bar{X}(s),\bar{X}(s)\rangle ds-\int_{t}^{T}\left\langle M^{-1}(s,s)S(s,s) \bar{X}(s),S(s,s)\bar{X}(s)\right\rangle ds\\
		&&+\langle G(t)\bar{X}(T),\bar{X}(T)\rangle+\int_{t}^{T}\left\langle M^{-1}(s,s)B^{\top}(s)P(s)\bar{X}(s),B^{\top}(s)P(s)\bar{X}(s)\right\rangle ds\\
		&&-\int_{t}^{T}\left\langle M^{-1}(s,s)\left(B^{\top}(s)P(s)+S(s,s)\right) \bar{X}(s),\left(B^{\top}(s)P(s)-S(s,s)\right)\bar{X}(s)\right\rangle  ds\\
		&&+\int_{t}^{T}\left\langle [M(t,s) \bar{u}(s,\bar{X}(s))+2S(t,s)\bar{X}(s)],\bar{u}(s,\bar{X}(s))\right\rangle ds\\
		&=&\langle G(t)\bar{X}(T),\bar{X}(T)\rangle+\int_{t}^{T}\langle Q(t,s)\bar{X}(s),\bar{X}(s)\rangle ds\\
		&&+\int_{t}^{T}\left\langle [M(t,s) \bar{u}(s,\bar{X}(s))+2S(t,s)\bar{X}(s)],\bar{u}(s,\bar{X}(s))\right\rangle ds
	\end{eqnarray*}
which yields that
	\begin{eqnarray}\label{2.28}
	J(t,x;\bar{u})=\langle P(t)x,x\rangle\mbox{ for any }(t,x)\in [0,T]\times\mathbb{R}^{n}.
	\end{eqnarray}

In order to verify $\hat{u}$ given by (\ref{2.24}) is the equilibrium control,	 we consider the perturbation control $u^{\varepsilon,v}(t,x)$ given by (\ref{1.4}). Solving the control system (\ref{1.1}) with $u^{\varepsilon,v}(t,x)$, we have the control systems (\ref{1.1}) with $u^{\varepsilon,v}$ has a unique solution $ X_{t,x}^{u^{\varepsilon,v}}$ given by
	\begin{equation} \label{2.29}
	X_{t,x}^{u^{\varepsilon,v}}(s)=\left\{\begin{array}{ll}
	X^{v}_{t,x}(s),&s\in [t,t+\varepsilon],\\
	\Phi (s,t+\varepsilon)X^{v}_{t,x}(t+\varepsilon),&s\in (t+\varepsilon,T].
	\end{array}
	\right.
	\end{equation}
	It follows from  (\ref{1.2}) and  (\ref{2.29}), (\ref{2.24}), (\ref{2.28})  that
	\begin{eqnarray*}
		&&J\left(t,x;u^{\varepsilon,v}\right)\\
		&=&\int_{t}^{t+\varepsilon}\bigg[\left\langle Q(t,s)X^{v}_{t,x}(s),X^{v}_{t,x}(s)\right\rangle+2\left\langle S(t,s)X^{v}_{t,x}(s),v\right\rangle+\left\langle M(t,s)v,v\right\rangle\bigg]ds\\
		&&+\int_{t+\varepsilon}^{T}\left\langle (Q(t,s)-Q(t+\varepsilon,s))\Phi (s,t+\varepsilon)X^{v}_{t,x}(t+\varepsilon),\Phi (s,t+\varepsilon)X^{v}_{t,x}(t+\varepsilon)\right\rangle ds\\
		&&-2\int_{t+\varepsilon}^{T}\bigg\langle   (S(t,s)-S(t+\varepsilon))\Phi (s,t+\varepsilon)X^{v}_{t,x}(t+\varepsilon),\\
		&&M^{-1}(s,s)\left(B^{\top}(s)P(s)+S(s,s)\right)\Phi (s,t+\varepsilon)X^{v}_{t,x}(t+\varepsilon)\bigg\rangle ds\\
		&&+\int_{t+\varepsilon}^{T}\bigg \langle  (M(t,s)-M(t+\varepsilon,s))M^{-1}(s,s)\left(B^{\top}(s)P(s)+S(s,s)\right)\Phi (s,t+\varepsilon)X^{v}_{t,x}(t+\varepsilon),\\
		&&M^{-1}(s,s)\left(B^{\top}(s)P(s)+S(s,s)\right)\Phi (s,t+\varepsilon)X^{v}_{t,x}(t+\varepsilon)\bigg\rangle ds\\
		&&+\left\langle (G(t)-G(t+\varepsilon))\Phi (T,t+\varepsilon)X^{v}_{t,x}(t+\varepsilon),\Phi (T,t+\varepsilon)X^{v}_{t,x}(t+\varepsilon)\right\rangle\\
		&&+\left\langle  P(t+\varepsilon)X^{v}_{t,x}(t+\varepsilon),X^{v}_{t,x}(t+\varepsilon)\right\rangle.
	\end{eqnarray*}
Moreover, it follows from (\ref{2.29})  that $\lim\limits_{\varepsilon \rightarrow0}	X_{t,x}^{u^{\varepsilon,v}}(\cdot)=\Phi (\cdot,t)x $ in $ C\left([t,T];\mathbb{R}^n\right)$, then we have
	\begin{eqnarray*}
		&&\lim_{\varepsilon\searrow 0} \frac{J\left(t,x;u^{\varepsilon,v}\right)- J\left(t,x;\bar{u}\right)}{\varepsilon}\nonumber\\
		&=&\left\langle Q(t,t)x,x\right\rangle+2\left\langle S(t,t)x,v\right\rangle+\left\langle M(t,t)v,v\right\rangle-\int_{t}^{T}\left\langle Q_t(t,s) \Phi (s,t )x,\Phi (s,t )x\right\rangle ds\\
		&&+2\int_{t}^{T}\bigg\langle  S_t(t,s) \Phi (s,t )x,M^{-1}(s,s)\left(B^{\top}(s)P(s)+S(s,s)\right)\Phi (s,t ) x\bigg\rangle ds\\
		&&-\int_{t }^{T}\bigg \langle M_t(t,s)M^{-1}(s,s)\left(B^{\top}P(s)+S(s,s)\right)\Phi (s,t)x,\\
		&&M^{-1}(s,s)\left(B^{\top}(s)P(s)+S(s,s)\right)\Phi (s,t)x\bigg\rangle ds+2\left\langle B^{\top}(t)P(t)x,v\right\rangle\\
		&&-\left\langle \dot{G}(t)\Phi (T,t)x,\Phi (T,t )x\right\rangle+\left\langle \left( \dot{P}(t)+P(t)A(t)+A^{\top}(t)P(t)\right)x,x\right\rangle.
	\end{eqnarray*}
	which, together with (\ref{1.8}), yields that
	\begin{eqnarray*}\label{2.30}
	&&\lim_{\varepsilon\searrow 0} \frac{J\left(t,x;u^{\varepsilon,v}\right)- J\left(t,x;\bar{u}\right)}{\varepsilon}\nonumber\\
	&=&\left\langle \left( \dot{P}(t)+P(t)A(t)+A^{\top}(t)P(t)+\bar{Q}(t,t)\right)x,x\right\rangle+2\left\langle \left(B^{\top}(t)P(t)+S(t,t)\right)x,v\right\rangle +\left\langle M(t,t)v,v\right\rangle\nonumber\\
	&=&\left\langle M^{-1}(t,t)\left[B^{\top}(t)P(t)+S(t,t)\right]x,\left[B^{\top}(t)P(t)+S(t,t)\right]x\right\rangle\\
	&&+2\left\langle \left[B^{\top}(t)P(t)+S(t,t)\right]x,v\right\rangle +\left\langle M(t,t)v,v\right\rangle\nonumber\\
	&=&\left\langle M(t,t)\left[M^{-1}(t,t)\left(B^{\top}(t)P(t)+S(t,t)\right)x+v\right],  M^{-1}(t,t)\left(B^{\top}(t)P(t)+S(t,t)\right)x+v\right\rangle\nonumber\\
	&\geq&0, \quad \forall(t,x,v)\in [0,T)\times\mathbb{R}^{n}\times\mathbb{R}^{m}.\nonumber
	\end{eqnarray*}
This suggests that $\bar{u}$ is an equilibrium control and thus completing the proof.
\end{proof}
\section{Solvability of the time inconsistent LQ problem}\label{Sec:Riccati}
\subsection{Solvability of the equilibrium Riccati equation}
In this section, we will study the solvability of the equilibrium Riccati equation  (\ref{1.8}). Throughout this section, besides Assumptions (H1)-(H4), we work with the following assumption
\begin{itemize}
\item[(H5)] 
$Q_t(t,s)$, $M_t(t,s)$, $\dot{G}(t)$, $Q(t,s)-S^{\top}(t,s)M^{-1}(t,s)S(t,s)$ and $Q_t(t,s)-S_t^{\top}(t,s)M_t^{-1}(t,s)S_t(t,s)$  are positive semi-definite, $\forall 0\leq t\leq  s \leq T$.
\end{itemize}
\begin{remark}
In the one dimensional case with $S\equiv 0$, let us suppose that $Q(t,s) =Q(s-t), M(t,s) = M(s-t), G(T-t) = G(t)$ and view $Q, M, G$ as discount functions. Then Assumption (H5) states the decreasing property of discount functions, which is commonly assumed in economics. 
\end{remark}
We define a map $\mathbb{F}$ as follows:
\begin{eqnarray}\label{7.2}
\mathbb{F}(t;s,P)&=&\Phi^{\top}(T,t)\dot{G}(s)\Phi(T,t)+\int_{t}^{T}\Phi^{\top}(\tau,t)Q_s(s,\tau)\Phi(\tau,t)d\tau\\
&&+\int_{t}^{T}\Phi^{\top}(\tau,t)\bigg[\Upsilon^{\top}(\tau)M_s(s,\tau) \Upsilon(\tau)-\Upsilon^{\top}(\tau)S_s(s,\tau)-S^{\top}_s(s,\tau)\Upsilon(\tau)\bigg]\Phi(\tau,t)d\tau,\;\;\forall s,t\in[0,T],\nonumber
\end{eqnarray}
It is easy to see that $\mathbb{F}(t;s,P)$ is symmetric and
\begin{eqnarray*}\label{7.3}
\bar{Q}(t,t)=Q(t,t)-\mathbb{F}(t;t,P)\mbox{ for all }t\in [0,T].
\end{eqnarray*}
\begin{lemma}\label{peng-l5.1} Suppose $P_1,P_2\in C\left([0,T];\mathbb{R}^{n\times n}\right)$ are symmetric matrix-value functions, then
\begin{eqnarray}\label{7.6}
\|\mathbb{F}(s;t,P_2)-\mathbb{F}(s;t,P_1)\|\leq 4T\gamma e^{4\omega}
\|P_1-P_2\|_{C}\mbox{ for any }s,t\in [0,T],
\end{eqnarray}
 where
\begin{eqnarray}\label{7.7}
\left\{\begin{array}{ll}
\omega=\|A\|_{L^1}+T\|B\|_{L^{\infty}}\|M^{-1}\|_{C}\left(\|B\|_{L^{\infty}}\left(\|P_1\|_{C}+\|P_2\|_{C}\right)+\|S\|_C\right),\\
\alpha=\|M^{-1}\|_{C}\left(\|B\|_{L^{\infty}}\left(\|P_1\|_{C}+\|P_2\|_{C}\right)+\|S\|_C\right),\\
\gamma=\left\|M^{-1}\right\|_C\left(1+\|B\|_{L^{\infty}}\right)^2\bigg[\|G\|_{C^1}+T\|Q\|_{C^1}+
(1+T\alpha)\left(\alpha\|M\|_{C^1}+\|S\|_{C^1}\right)\bigg].
\end{array}\right.
\end{eqnarray}
\end{lemma}
\begin{proof}
Define 
\begin{eqnarray}\label{7.8}
\Psi_i(t,s)=\exp\left(\int_{s}^{t}\left(A(\tau)-B(\tau)M^{-1}(\tau,\tau)\left(B^{\top}(\tau)P_i(\tau)+S(\tau,\tau)\right)\right)d\tau\right),\;\;s,t\in[0,T], i=1,2.
\end{eqnarray}
Then we have
\begin{eqnarray}\label{7.10}
\|\Psi_i(t,s)\|\leq\exp\left(\|A\|_{L^1}+T\alpha\|B\|_{L^{\infty}}\right),\;\;s,t\in[0,T],\;\; i=1,2.
\end{eqnarray}
Thus, it follows from (\ref{7.2}) and (\ref{7.7}) that
\begin{eqnarray}\label{7.11}
\|\mathbb{F}(t;s,P_1)\|&\leq&\|\Psi_1(T,t)\|^2\|\dot{G}(s)\|+\int_{t}^{T}\|\Psi_1(\tau,t)\|^2\bigg[\|Q_s(s,\tau)\|
+\bigg\|\bigg(B^{\top}(\tau)P_1(\tau)\nonumber\\
&&+S(\tau,\tau)\bigg)^{\top}M^{-1}(\tau,\tau)\left(M_s(s,\tau) M^{-1}(\tau,\tau)\left(B^{\top}(\tau)P_1(\tau)+S(\tau,\tau)\right)-2S_s(s,\tau)\right)\bigg\|\bigg]d\tau\nonumber\\
&\leq&\left[\|G\|_{C^1}+T\|Q\|_{C^1}+T\alpha\left(\alpha\|M\|_{C^1}+2\|S\|_{C^1}\right)\right]\exp\left(2\|A\|_{L^1}+2T\alpha\|B\|_{L^{\infty}}\right).
\end{eqnarray}
For any fixed  $s\in[0,T]$, we define
\begin{eqnarray*}\label{7.12}
\mathbb{E}(t)=\mathbb{F}(t;s,P_2)-\mathbb{F}(t;s,P_1)\mbox{ for all }t\in[0,T].
\end{eqnarray*}
Then,  some algebra yields that $\mathbb{E}\in C\left([0,T];\mathbb{R}^{n\times n}\right)\bigcap W^{1,1}\left((0,T);\mathbb{R}^{n\times n}\right)$ satisfies the following lyapunov  equation
\begin{eqnarray*}
\left\{\begin{array}{ll}
\dot{\mathbb{E}}(t)+\left(A(t)-B(t)\Upsilon_2(t)\right)^{\top}\mathbb{E}(t)+\mathbb{E}(t)\left(A(t)-B(t)\Upsilon_2(t)\right)\\
\qquad+\left(P_1(t)-P_2(t)\right)^{\top}B(t)M^{-1}(t,t)\left[B^{\top}(t)\mathbb{F}(t;s,P_1)-M_s(s,t) \Upsilon_1(t)+S_s(s,t)\right]\\
\qquad+\left[\mathbb{F}(t;s,P_1)B(t)-\Upsilon^{\top}_2(t)M_s(s,t)+S^{\top}_s(s,t)\right]M^{-1}(t,t)B^{\top}(t)\left(P_1(t)-P_2(t)\right)=0,&t\in[0,T),\\
\mathbb{E}(T)=0,
\end{array}\right.
\end{eqnarray*}
where $\Upsilon_i(t)=M^{-1}(t,t)\left(B^{\top}(t)P_i(t)+S(t,t)\right)$, $i=1$, $2$. \par
Solving the  above lyapunov  equation with  (\ref{7.8}), we have that
\begin{eqnarray*}
\mathbb{E}(t)&=&\int_{t}^{T}\Psi_2^{\top}(\tau,t)\bigg[\left[\mathbb{F}(\tau;s,P_1)B(\tau)-\Upsilon^{\top}_2(\tau)M_s(s,\tau)+S^{\top}_s(s,\tau)\right]M^{-1}(\tau,\tau)B^{\top}(\tau)\left(P_1(\tau)-P_2(\tau)\right)\\
&&+\left(P_1(\tau)-P_2(\tau)\right)^{\top}B(\tau)M^{-1}(\tau,\tau)\left(B^{\top}(\tau)\mathbb{F}(\tau;s,P_1)-M_s(s,\tau) \Upsilon_1(t)+S_s(s,\tau)\right)\bigg]\Psi_2(\tau,t)d\tau.\nonumber
\end{eqnarray*}
As $\mathbb{E}(t)=\mathbb{F}(t;s,P_2)-\mathbb{F}(t;s,P_1)$, it follows from (\ref{7.7}),  (\ref{7.10}) and (\ref{7.11}) that
\begin{eqnarray*}
\left\|\mathbb{F}(t;s,P_2)-\mathbb{F}(t;s,P_1)\right\|&\leq&2Te^{2\omega}\|B\|_{L^{\infty}}\left\|M^{-1}\right\|_C\bigg[\|\mathbb{F}(\cdot;s,P_1)\|_C\|B\|_{L^{\infty}}+\alpha\|M\|_{C^1} +\|S\|_{C^1}\bigg] \|P_2-P_1\|_C\nonumber\\
&\leq&4Te^{4\omega}\left\|M^{-1}\right\|_C\left(1+\|B\|_{L^{\infty}}\right)^2\bigg[\|G\|_{C^1}+T\|Q\|_{C^1}\nonumber\\
&&+
(1+T\alpha)\left(\alpha\|M\|_{C^1}+\|S\|_{C^1}\right)\bigg] \|P_2-P_1\|_C\nonumber\\
&\leq&4\gamma Te^{4\omega} \|P_2-P_1\|_C\nonumber
\end{eqnarray*}
This completes the proof.
\end{proof}\par 
The following lemma, which demonstrates different types of solutions to the Riccati equation (\ref{1.8}), will be used in the proof of the main result of this section.
\begin{lemma}\label{peng-l5.2}  
Suppose  $P(\cdot)\in C\left([0,T];\mathbb{R}^{n\times n}\right)$ is symmetric. Then the following statements are equivalent.
\begin{itemize}
\item[(i)]  $P$ solves the Riccati differential equation (\ref{1.8}).
\item[(ii)]  $P$ solves the following Riccati equation 
\begin{eqnarray*}\label{RpP1}
P(t)&=&\Psi^{\top}(T,t)G(T)\Psi(T,t)+\int_t^T\Psi^{\top}(s,t)\bigg[Q(s,s)-\mathbb{F}(s;s,P)\\
&&-\left(B^{\top}(s)P(s)+S(s,s)\right)^{\top}M^{-1}(s,s)\left(B^{\top}(s)P(s)+S(s,s)\right)\bigg]\Psi(s,t)ds,\;\;t\in[0,T]\nonumber
\end{eqnarray*}
where
\begin{eqnarray*}\label{Psi}
\Psi(s,t)=\exp\left(\int_t^s A(\tau)d\tau\right)\mbox{ for all }s,t\in[0,T].
\end{eqnarray*}
\item[(iii)]  $P$ solves the following Riccati integral equation
\begin{eqnarray*}\label{RpP2}
P(t)&=&\Phi^{\top}(T,t)G(T)\Phi(T,t)+\int_t^T\Phi^{\top}(s,t)\bigg[P(s)B(s)M^{-1}(s,s)B^{\top}(s)P(s)\\
&&-\mathbb{F}(s;s,P)+Q(s,s)-S^{\top}(s,s)M^{-1}(s,s)S(s,s)\bigg]\Phi(s,t)ds,\;\;t\in[0,T].\nonumber
\end{eqnarray*}
\end{itemize}
\end{lemma}
\begin{proof}
The result is an immediate consequence of direct calculation. Indeed, let\\ $Y(t):=(\Psi^{-1})^T(T,t)P(t)\Psi^{-1}(T,t)$. Then the the equivalence between (i) and (ii) follows from taking the first order derivative on $Y$ and the Riccati equation (\ref{1.8}). Similarly, let $Z(t):=(\Phi^{-1})^T(T,t)P(t)\Phi^{-1}(T,t)$. Then the the equivalence between (i) and (iii) follows from taking the first order derivative on $Z$ and the Riccati equation (\ref{1.8}). This completes the proof.
\end{proof}\par 
The following Lemma provides a prior estimate for the solution to the Riccati equation.\par 
\begin{lemma}\label{peng-l5.3} Suppose $P(\cdot)\in C\left([0,T];\mathbb{R}^{n\times n}\right)$ is symmetric and solves the Riccati equation (\ref{1.8}), then
$\mathbb{F}(\cdot;s,P)$ and $P$ are positive semi-definite and
\begin{eqnarray}\label{7.14}
\| P\|_C
\leq  e^{2\|A\|_{L^1}}\left( \|G(T)\|+T\|Q\|_C\right).
\end{eqnarray}
\end{lemma}
\begin{proof}
It follows from  the assumption (H5) and  (\ref{7.2}) that
\begin{eqnarray}\label{7.15}
&&\langle\mathbb{F}(t;s,P)x,x\rangle\nonumber\\
&=&\langle\dot{G}(s)\Phi(T,t)x,\Phi(T,t)x\rangle+\int_{t}^{T}\left\langle \left[Q_s(s,\tau)-2\Upsilon^{\top}(\tau)S_s(s,\tau)+\Upsilon^{\top}(\tau)M_s(s,\tau) \Upsilon(\tau)\right]\Phi(\tau,t)x,\Phi(\tau,t)x\right\rangle d\tau\nonumber\\
&=&\langle\dot{G}(s)\Phi(T,t)x,\Phi(T,t)x\rangle+\int_{t}^{T}\left\langle \left(Q_s(s,\tau)-2S^{\top}_s(s,\tau)M^{-1}_s(s,\tau)S_s(s,\tau)\right)\Phi(\tau,t)x,\Phi(\tau,t)x\right\rangle d\tau\nonumber\\
&&+\int_{t}^{T}\left\langle M_s(s,\tau)\left[\Upsilon(\tau)-M^{-1}_s(s,\tau)S_s(s,\tau)\right]\Phi(\tau,t)x,\left[\Upsilon(\tau)-M^{-1}_s(s,\tau)S_s(s,\tau)\right]\Phi(\tau,t)x\right\rangle d\tau\nonumber\\
&\geq&0 \mbox{ for all }\;\;(s,t,x)\in[0,T]\times [0,T]\times \mathbb{R}^n.
\end{eqnarray}
Moreover, for any $x\in \mathbb{R}^n$,  using Lemma~\ref{peng-l5.2}, we have
\begin{eqnarray*}
&&\left\langle P(t)x,x\right\rangle \\
&=&\left\langle G(T)\Phi(T,t)x,\Phi(T,t)x\right\rangle -\int_t^T\left\langle  \dot{G}(s)\Phi(T,t)x,\Phi(T,t)x\right\rangle ds\nonumber\\
&&+\int_t^T\left\langle \left[P(s)B(s)M^{-1}(s,s)B^{\top}(s)P(s)+Q(s,s)-S^{\top}(s,s)M^{-1}(s,s)S(s,s)\right]\Phi(s,t)x,\Phi(s,t)x\right\rangle ds\nonumber\\
&&-\int_t^T\int_{s}^{T}\left\langle \left[Q_s(s,\tau)+\Upsilon^{\top}(\tau)M_s(s,\tau) \Upsilon(\tau)-2\Upsilon^{\top}(\tau)S_s(s,\tau)\right]\Phi(\tau,t)x,\Phi(\tau,t)x\right\rangle d\tau ds\\
&=&\left\langle G(t)\Phi(T,t)x,\Phi(T,t)x\right\rangle +\int_t^T\left\langle \left[Q(t,s)-2\Upsilon^{\top}(s)S(t,s)+\Upsilon^{\top}(s)M(t,s)\Upsilon(\tau)\right]\Phi(s,t)x,\Phi(s,t)x\right\rangle ds \\
&=&\left\langle G(t)\Phi(T,t)x,\Phi(T,t)x\right\rangle +\int_t^T\left\langle \left[Q(t,s)-S^{\top}(t,s)M^{-1}(t,s)S(t,s)\right]\Phi(s,t)x,\Phi(s,t)x\right\rangle ds \\
&&+\int_t^T\left\langle M(t,s)\left(\Upsilon(s)-M^{-1}(t,s)S(t,s)\right)\Phi(s,t)x,\left(\Upsilon(s)-M^{-1}(t,s)S(t,s)\right)\Phi(s,t)x\right\rangle ds.
\end{eqnarray*}
Thanks to (H2), (H3) and (H5), we thus have that
\begin{eqnarray}\label{7.16}
\left\langle P(t)x,x\right\rangle\geq 0\mbox{ for all }(t,x)\in [t_{12},T]\times\mathbb{R}^n.
\end{eqnarray}
Moreover, it follows from (\ref{7.15}) that
\begin{eqnarray*}
\langle P(t)x,x\rangle&=&\langle G(T)\Psi(T,t)x,\Psi(T,t)x\rangle+\int_t^T\bigg\langle\bigg[Q(s,s)-\mathbb{F}(s;s,P)\nonumber\\
&&-\left(B^{\top}(s)P(s)+S(s,s)\right)^{\top}M^{-1}(s,s)\left(B^{\top}(s)P(s)+S(s,s)\right)\bigg]\Psi(s,t)x,\Psi(s,t)x\bigg\rangle ds\nonumber\\
&\leq&\langle G(T)\Psi(T,t)x,\Psi(T,t)x\rangle+\int_t^T \langle Q(s,s)\Psi(s,t)x,\Psi(s,t)x \rangle ds,\;\;\forall(t,x)\in[0,T]\times \mathbb{R}^n
\end{eqnarray*}
which, together with (\ref{7.16}), yields (\ref{7.14}) and thus completing the proof.
\end{proof}\par 

The main theorem in this section demonstrates the unique solvability of the Riccati equation.
\begin{theorem}\label{peng-theorem5.1} The equilibrium Riccati equation (\ref{1.8}) admits a unique symmetric solution $P(\cdot)\in C\left([0,T];\mathbb{R}^{n\times n}\right)$. Moreover, $P$ is positive semi-definite.
\end{theorem}
\begin{proof}
The proof can be divided into two steps. In Step $1$, we prove the local existence and uniqueness of the solution to the Riccati equation (\ref{1.8}), while we extend the result to the global case in Step 2.\par 
For the ease of exposition, we introduce the following notations, which will be used in the proof without further clarifications.\par 
 Denote by $\mathbb{B}(\eta,r_0)$ the closed ball (in $\mathbb{R}^{n\times n}$) centered at $\eta$ and of radius $r_0 >0$. Define 
\begin{eqnarray}\label{7.18}
\left\{\begin{array}{ll}
r=e^{2\|A\|_{L^1}}\left( \|G\|_C+T\|Q\|_C\right),\\
\bar{\rho}=\|M^{-1}\|_{C}\left(4r\|B\|_{L^{\infty}}+\|S\|_C\right),\\
\bar{\beta}=\|A\|_{L^1}+T\bar{\rho}\|B\|_{L^{\infty}},\\
\bar{\omega}=\|A\|_{L^1}+2T\bar{\rho}\|B\|_{L^{\infty}},\\
\bar{\gamma}=\left\|M^{-1}\right\|_C\left(1+\|B\|_{L^{\infty}}\right)^2\bigg[\|G\|_{C^1}+T\|Q\|_{C^1}+
(1+2T\bar{\rho})\left(2\bar{\rho}\|M\|_{C^1}+\|S\|_{C^1}\right)\bigg].
\end{array}\right.
\end{eqnarray}
For $\Psi$,  we can find a constant $\tau_1>0$ such that
\begin{eqnarray}\label{7.21}
\|\Psi(t,s)-I\|\leq \frac{1}{2\left(1+e^{2\bar{\beta}}\right)} \mbox{ for any } s,t\in [0,T] \mbox{ with } |s-t|\leq \tau_1,
\end{eqnarray}
where the existence of $\tau_1$ follows from the uniform continuity of $\Psi.$\par 
Moreover, we define
\begin{eqnarray}\label{7.22}
\left\{\begin{array}{ll}
\tau_2=\frac{r}{2e^{4\bar{\beta}}\left(\bar{\rho}^2\left\|M\right\|_C+\|G\|_{C^1}+T\|Q\|_{C^1}
+T\bar{\rho}\left(\bar{\rho}\|M\|_{C^1} +2\|S\|_{C^1}\right)+\|Q\|_C\right)},\\
\tau_3=\frac{1}{4e^{2\|A\|_{L^1}}\left(\bar{\rho}\|B\|_{L^{\infty}}
+2T\bar{\gamma} e^{4\bar{\omega}}\right)}
\end{array}\right.
\end{eqnarray}
and
\begin{eqnarray}\label{7.23}
\tau=\min\{\tau_1,
\tau_2,
\tau_3\}.
\end{eqnarray}
It is easy to see that $\tau>0.$\par 
\textit{Step $1$}.
Define a map $\mathbb{H}_1$ on $C([T-\tau,T];\mathbb{B}(G(T),r))$ given by
\begin{eqnarray}\label{7.24}
(\mathbb{H}_1P)(t)&=&\Psi^{\top}(T,t)G(T)\Psi(T,t)+\int_t^T\Psi^{\top}(s,t)\bigg[Q(s,s)-\mathbb{F}(s;s,P)\\
&&-\left(B^{\top}(s)P(s)+S(s,s)\right)^{\top}M^{-1}(s,s)\left(B^{\top}(s)P(s)+S(s,s)\right)\bigg]\Psi(s,t)ds.\nonumber
\end{eqnarray}
Following (\ref{7.11}) and using the notations defined by (\ref{7.18})- (\ref{7.23}), we have 
\begin{eqnarray*}\label{7.20}
\|\mathbb{F}(t;s,P)\|\leq e^{2\bar{\beta}}\left[\|G\|_{C^1}+T\|Q\|_{C^1}
+T\bar{\rho}\left(\bar{\rho}\|M\|_{C^1} +2\|S\|_{C^1}\right)\right]\mbox{ for any }P\in C([0,T];\mathbb{B}(0,2r)).
\end{eqnarray*}
Thus,
\begin{eqnarray}\label{7.25}
&&\|(\mathbb{H}_1P)(t)-G(T)\|\nonumber\\
&=&\bigg\|(\Psi(T,t)-I)^{\top}G(T)\Psi(T,t)+G(T)(\Psi(T,t)-I)+\int_t^T\Psi^{\top}(s,t)\bigg[Q(s,s)-\mathbb{F}(s;s,P)\nonumber\\
&&-\left(B^{\top}(s)P(s)+S(s,s)\right)^{\top}M^{-1}(s,s)\left(B^{\top}(s)P(s)+S(s,s)\right)\bigg]\Psi(s,t)ds\bigg\|\nonumber\\
&\leq&\left(1+e^{2\bar{\beta}}\right)\|G(T)\|\|\Psi(T,t)-I\|+\int_t^T\|\Psi\|_C^2\bigg[\|Q\|_C+\|\mathbb{F}(s;s,P)\|\nonumber\\
&&+\left\|B^{\top}(s)P(s)+S(s,s)\right\|^2\left\|M^{-1}(s,s)\right\|\bigg]ds\nonumber\\
&\leq & e^{4\bar{\beta}}\left(\bar{\rho}^2\left\|M\right\|_C+\|G\|_{C^1}+T\|Q\|_{C^1}
+T\bar{\rho}\left(\bar{\rho}\|M\|_{C^1} +2\|S\|_{C^1}\right)+\|Q\|_C\right)(T-t)\nonumber\\
&&+\left(1+e^{2\bar{\beta}}\right)\|G(T)\|\|\Psi(T,t)-I\|\nonumber\\
&\leq& r
\end{eqnarray}
for all $P\in C([T-\tau,T];\mathbb{B}(G(T),r))$.\par 
This shows that $\mathbb{H}_1C([T-\tau,T];\mathbb{B}(G(T),r))\subseteq C([T-\tau,T];\mathbb{B}(G(T),r))$.\par 
Furthermore, for any $P_1,P_2\in C([T-\tau,T];\mathbb{B}(G(T),r))$, it is easy to see from (\ref{7.24}) that
\begin{eqnarray*}
(\mathbb{H}_1P_1)(t)-(\mathbb{H}_1P_2)(t)&=&\int_t^T\Psi^{\top}(s,t)\bigg[(P_2(s)-P_1(s))B(s)M^{-1}(s,s)\left(B^{\top}(s)P_2(s)+S(s,s)\right)\nonumber\\
&&+\left(B^{\top}(s)P_1(s)+S(s,s)\right)^{\top}M^{-1}(s,s)B^{\top}(s)(P_2(s)-P_1(s))\\
&&+\mathbb{F}(s;s,P_2)-\mathbb{F}(s;s,P_1)\bigg]\Psi(s,t)ds.\nonumber
\end{eqnarray*}
As a result,
\begin{eqnarray}\label{7.26}
\|(\mathbb{H}_1P_1)(t)-(\mathbb{H}_1P_2)(t)\|&\leq&\int_t^T\|\Psi(s,t)\|^2\bigg[2\left\|B^{\top}(s)P_1(s)+S(s,s)\right\|\left\|M^{-1}(s,s)\right\|\|B(s)\|\|P_2(s)-P_1(s)\|\nonumber\\
&&+\|\mathbb{F}(s;s,P_2)-\mathbb{F}(s;s,P_1)\|\bigg]ds.
\end{eqnarray}
It follows from (\ref{7.6}) that
\begin{eqnarray}\label{7.19}
\|\mathbb{F}(s;t,P_2)-\mathbb{F}(s;t,P_1)\|\leq 4T\bar{\gamma} e^{4\bar{\omega}}
\|P_1-P_2\|_{C}\mbox{ for any }s,t\in [0,T],P_1,P_2\in C([0,T];\mathbb{B}(0,2r))
\end{eqnarray}
Combining (\ref{7.26}) and (\ref{7.19}) and the notations defined in (\ref{7.22}), (\ref{7.23}), we obtain
\[
\|(\mathbb{H}_1P_1)(t)-(\mathbb{H}_1P_2)(t)\|\leq  \frac{1}{2}\|P_2-P_1\|_{C}\mbox{ for any }P_1,P_2\in C([T-\tau,T];\mathbb{B}(G(T),r)).
\]
The above inequality and (\ref{7.25}) shows that $\mathbb{H}_1$ is a contraction mapping from $C([T-\tau,T];\mathbb{B}(G(T),r))$ to itself. Therefore, it follows from the Banach fixed point theorem that, there is a unique $\mathbb{P}_1\in C([T-\tau,T];\mathbb{B}(G(T),r))$ such that
\begin{eqnarray*}\label{7.27}
\mathbb{P}_1(t)&=&\Psi^{\top}(T,t)G(T)\Psi(T,t)+\int_t^T\Psi^{\top}(s,t)\bigg[Q(s,s)-\mathbb{F}(s;s,\mathbb{P}_1)\\
&&-\left(B^{\top}(s)\mathbb{P}_1(s)+S(s,s)\right)^{\top}M^{-1}(s,s)\left(B^{\top}(s)\mathbb{P}_1(s)+S(s,s)\right)\bigg]\Psi(s,t)ds, \;\;\forall t\in [T-\tau,T].\nonumber
\end{eqnarray*}
Hence, according to Lemma ~\ref{peng-l5.2}, we have that $\mathbb{P}_1$ solves the Riccati equation (\ref{1.8}) on $[T-\tau,T]$, and the positive semi-definiteness of $\mathbb{P}_1$ follows from Lemma~\ref{peng-l5.3}.\par
\textit{Step $2$.}
For any fixed $P\in C\left([T-2\tau,T-\tau];\mathbb{R}^{n\times n}\right)$, define
\begin{eqnarray}\label{7.29}
\left\{\begin{array}{ll}
\hat{P}_1(t)=\left\{\begin{array}{ll}
P(t),&t\in [T-2\tau, T-\tau),\\
\mathbb{P}_1(t),&t\in [T-\tau,T],
\end{array}\right.\\
\hat{\Upsilon}_1(t)=M^{-1}(t,t)\left(B^{\top}(t)\hat{P}_1(t)+S(t,t)\right),\\
\hat{\Phi}_1(t,s)=\exp\left(\int_{s}^{t}\left(A(\tau)-B(\tau)M^{-1}(\tau,\tau)\left(B^{\top}(\tau)\hat{P}_1(\tau)+S(\tau,\tau)\right)\right)d\tau\right),~~\forall s,t\in[T-2\tau,T],
\end{array}\right.
\end{eqnarray}
 and
\begin{eqnarray}\label{7.30}
\mathbb{F}_1(t;s,\hat{P}_1)&=&\hat{\Phi}_1^{\top}(T,t)\dot{G}(s)\hat{\Phi}_1(T,t)+\int_{t}^{T}\hat{\Phi}_1^{\top}(\tau,t)Q_s(s,\tau)\hat{\Phi}_1(\tau,t)d\tau\\
&&+\int_{t}^{T}\hat{\Phi}_1^{\top}(\tau,t)\bigg[\hat{\Upsilon}_1^{\top}(\tau)M_s(s,\tau) \hat{\Upsilon}_1(\tau)-\hat{\Upsilon}_1^{\top}(\tau)S_s(s,\tau)-S^{\top}_s(s,\tau)\hat{\Upsilon}_1(\tau)\bigg]\hat{\Phi}_1(\tau,t)d\tau\nonumber
\end{eqnarray}
for $s,t\in[T-2\tau,T]$.\par
 We introduce a map $\mathbb{H}_2$ on $C([T-2\tau,T-\tau];\mathbb{B}(\mathbb{P}_1(T-\tau),r))$ given by
\begin{eqnarray*}\label{7.31}
(\mathbb{H}_2P)(t)&=&\Psi^{\top}(T-\tau,t)\mathbb{P}_1(T-\tau)\Psi(T-\tau,t)+\int_t^{T-\tau}\Psi^{\top}(s,t)\bigg[Q(s,s)-\mathbb{F}_1(s;s,\hat{P}_1)\\
&&-\left(B^{\top}(s)P(s)+S(s,s)\right)^{\top}M^{-1}(s,s)\left(B^{\top}(s)P(s)+S(s,s)\right)\bigg]\Psi(s,t)ds.\nonumber
\end{eqnarray*}
Moreover, combining (\ref{7.18}), (\ref{7.29}) and (\ref{7.30}), we  have
\begin{eqnarray*}\label{7.32}
&&\|(\mathbb{H}_2P)(t)-\mathbb{P}_1(T-\tau)\|\nonumber\\
&\leq&\|(\Psi(T-\tau,t)-I)^{\top}\mathbb{P}_1(T-\tau)\Psi(T-\tau,t)+\mathbb{P}_1(T-\tau)(\Psi(T-\tau,t)-I)\|\nonumber\\
&&+\int_t^{T-\tau}\|\Psi(s,t)\|^2\bigg[\|Q(s,s)\|+\|\mathbb{F}_1(s;s,\hat{P}_1)\|+\left\|B^{\top}(s)P(s)+S(s,s)\right\|^2\left\|M^{-1}(s,s)\right\|\bigg]ds\nonumber\\
&\leq & e^{4\bar{\beta}}\bigg(\bar{\rho}^2\left\|M\right\|_C+\|G\|_{C^1}+\|Q\|_C+T\|Q\|_{C^1}
+T\bar{\rho}\left(\bar{\rho}\|M\|_{C^1} +2\|S\|_{C^1}\right)\bigg)(T-\tau-t)\\
&&+\left(1+e^{2\bar{\beta}}\right)\|\mathbb{P}_1(T-\tau)\|\|\Psi(T-\tau,t)-I\|\nonumber
\end{eqnarray*}
and
\begin{eqnarray*}\label{7.33}
(\mathbb{H}_2P_1)(t)-(\mathbb{H}_2P_2)(t)&=&\int_t^T\Psi^{\top}(s,t)\bigg[(P_2(s)-P_1(s))B(s)M^{-1}(s,s)\left(B^{\top}(s)P_2(s)+S(s,s)\right)\nonumber\\
&&+\left(B^{\top}(s)P_1(s)+S(s,s)\right)^{\top}M^{-1}(s,s)B^{\top}(s)(P_2(s)-P_1(s))\\
&&+\mathbb{F}(s;s,\hat{P}_2)-\mathbb{F}(s;s,\hat{P}_1)\bigg]\Psi(s,t)ds\nonumber
\end{eqnarray*}
for all $P_1, P_2\in C([T-2\tau,T-\tau];\mathbb{B}(\mathbb{P}_1(T-\tau),r))$.\par 
 It follows from  (\ref{7.18}), (\ref{7.23}),  (\ref{7.29}) and (\ref{7.30})
that $\mathbb{H}_2C([T-2\tau,T-\tau];\mathbb{B}(\mathbb{P}_1(T-\tau),r))\subseteq C([T-2\tau,T-\tau];\mathbb{B}(\mathbb{P}_1(T-\tau),r))$ and
\[
\|(\mathbb{H}_2P_1)(t)-(\mathbb{H}_2P_2)(t)\|\leq  \frac{1}{2}\|P_2-P_1\|_{C}\mbox{ for any }P_1,P_2\in C([T-2\tau,T-\tau];\mathbb{B}(\mathbb{P}_1(T-\tau),r)).
\]
Thus, it follows from the Banach fixed point theorem that, there is a unique $\mathbb{P}_2\in C([T-2\tau,T-\tau];\mathbb{B}(\mathbb{P}_1(T-\tau),r))$  such that
\begin{eqnarray}\label{7.34}
\mathbb{P}_2(t)&=&\Psi^{\top}(T-\tau,t)\mathbb{P}_1(T-\tau)\Psi(T-\tau,t)+\int_t^{T-\tau}\Psi^{\top}(s,t)\bigg[Q(s,s)-\mathbb{F}(s;s,\hat{\mathbb{P}}_2)\\
&&-\left(B^{\top}(s)\mathbb{P}_2(s)+S(s,s)\right)^{\top}M^{-1}(s,s)\left(B^{\top}(s)\mathbb{P}_2(s)+S(s,s)\right)\bigg]\Psi(s,t)ds, \;\;\forall t\in [T-2\bar{t},T-\tau],\nonumber
\end{eqnarray}
It is easy to see from (\ref{7.34}) that $\mathbb{P}_2$ is symmetric. Also, Lemma \ref{peng-l5.2} yields that $\mathbb{P}_2$ solves the Riccati equation (\ref{1.8}) and Lemma  \ref{peng-l5.3} shows that $\mathbb{P}_2$ is positive semi-definite.\par
Continue the above process, we then obtain the symmetric matrix-valued function $\{\mathbb{P}_k|k=1,2,\cdots, \left[\frac{T}{\tau}\right], \left[\frac{T}{\tau}\right]+1\}$ given by
\begin{eqnarray*}\label{7.36}
\mathbb{P}_k(t)&=&\Psi^{\top}(T-(k-1)\tau,t)\mathbb{P}_{k-1}(T-(k-1)\tau)\Psi(T-(k-1)\tau,t)+\int_t^{T-(k-1)\tau}\Psi^{\top}(s,t)\bigg[Q(s,s)-\mathbb{F}(s;s,\hat{\mathbb{P}}_k)\nonumber\\
&&-\left(B^{\top}(s)\mathbb{P}_k(s)+S(s,s)\right)^{\top}M^{-1}(s,s)\left(B^{\top}(s)\mathbb{P}_k(s)+S(s,s)\right)\bigg]\Psi(s,t)ds
\end{eqnarray*}
for all $ t\in [T-k\tau,T-(k-1)\tau]$ and
\begin{eqnarray*}\label{7.37}
\mathbb{P}_j(t)&=&\Psi^{\top}(T-(j-1)\tau,t)\mathbb{P}_{j-1}(T-(j-1)\tau)\Psi(T-(j-1)\tau,t)+\int_t^{T-(j-1)\tau}\Psi^{\top}(s,t)\bigg[Q(s,s)-\mathbb{F}(s;s,\hat{\mathbb{P}}_j)\nonumber\\
&&-\left(B^{\top}(s)\mathbb{P}_j(s)+S(s,s)\right)^{\top}M^{-1}(s,s)\left(B^{\top}(s)\mathbb{P}_j(s)+S(s,s)\right)\bigg]\Psi(s,t)ds
\end{eqnarray*}
for all $ t\in [0,T-\left[\frac{T}{\tau}\right]\tau]$, where $j=\left[\frac{T}{\tau}\right]+1$.
Here
\begin{eqnarray*}\label{7.38}
\left\{\begin{array}{ll}
\hat{\mathbb{P}}_k(t)=\left\{\begin{array}{ll}
\mathbb{P}_k(t),&t\in [T-k\tau, T-(k-1)\tau),\\
\vdots\\
\mathbb{P}_1(t),&t\in [T-\tau,T],
\end{array}\right.\\
\hat{\Upsilon}_k(t)=M^{-1}(t,t)\left(B^{\top}(t)\hat{\mathbb{P}}_k(t)+S(t,t)\right),\\
\hat{\Phi}_k(t,s)=\exp\left(\int_{s}^{t}\left(A(\tau)-B(\tau)M^{-1}(\tau,\tau)\left(B^{\top}(\tau)\hat{\mathbb{P}}_k(\tau)+S(\tau,\tau)\right)\right)d\tau\right),~~\forall s,t\in[T-k\tau,T],
\end{array}\right.
\end{eqnarray*}
 and
\begin{eqnarray*}\label{7.39}
\mathbb{F}(t;s,\hat{\mathbb{P}}_k)&=&\hat{\Phi}_k^{\top}(T,t)\dot{G}(s)\hat{\Phi}_k(T,t)+\int_{t}^{T}\hat{\Phi}_k^{\top}(\tau,t)Q_s(s,\tau)\hat{\Phi}_k(\tau,t)d\tau\\
&&+\int_{t}^{T}\hat{\Phi}_k^{\top}(\tau,t)\bigg[\hat{\Upsilon}_k^{\top}(\tau)M_s(s,\tau) \hat{\Upsilon}_k(\tau)-\hat{\Upsilon}_k^{\top}(\tau)S_s(s,\tau) - \hat{\Upsilon}_k(\tau)S_s^{\top}(s,\tau)\bigg]\hat{\Phi}_k(\tau,t)d\tau,\;\;s,t\in[0,T].\nonumber
\end{eqnarray*}
We define
\begin{eqnarray*}\label{7.40}
P(t)=\left\{\begin{array}{ll}
\mathbb{P}_1(t),&t\in [T-\tau,T],\\
\vdots\\
\mathbb{P}_k(t),&t\in [T-k\tau, T-(k-1)\tau),\\
\vdots\\
\mathbb{P}_j(t),&t\in [0,T-\left[\frac{T}{\tau}\right]\tau).
\end{array}\right.
\end{eqnarray*}
Moreover, $P\in C\left([0,T];\mathbb{R}^{n\times n}\right)$ satisfies the Riccati integral equation (\ref{7.11}) and
 $P$ is symmetric positive semi-definite.
This implies that  the Riccati differential equation (\ref{1.8}) has a unique symmetric solution $P\in C\left([0,T];\mathbb{R}^{n\times n}\right)$.
  This completes the proof.
\end{proof}
\subsection{Existence and Uniqueness of the linear equilibrium}
Having obtained the equivalence result and the solvability of the equilibrium Riccati equation, we are poised to show the existence and uniqueness of the linear equilibrium for the time inconsistent LQ problem.
\begin{theorem} The time-inconsistent LQ problem in Definition \ref{DefEquilibrium} admits a unique linear equilibrium.
\end{theorem}
\begin{proof}
The existence of the linear equilibrium is an immediate result of Proposition~\ref{peng-proposition2.3}.\par 
To prove the uniqueness, we let $\bar{u}(t,x)=\tilde{u}(t)x$ for all $(t,x)\in [0,T]\times \mathbb{R}^n$ denote a linear equilibrium of the LQ problem.\par 
Then (\ref{2.11}) yields that
\begin{eqnarray} \label{pc1}
\tilde{u}(t)=-M^{-1}(t,t)\left(B^{\top}(t)\tilde{P}(t)+S(t,t)\right)\mbox{ for all }t\in [0,T].
\end{eqnarray}
It now suffices to prove the uniqueness of $\tilde{P}$.\par 
It follows from (\ref{2.4}) and (\ref{2.3}) that
\begin{eqnarray} \label{pc2}
\left\{\begin{array}{ll}
\tilde{\Phi} (s,t)=\exp\left(\int_{t}^{s}\left(A(\nu)-B(\nu)M^{-1}(t,t)\left(B^{\top}(\nu)\tilde{P}(\nu)+S(\nu,\nu)\right)\right)d\nu\right) ,\forall t,s\in [0,T],\\
\tilde{Q}(t,s)=\left(B^{\top}(s)\tilde{P}(s)+S(s,s)\right)^{\top}M^{-1}(s,s)M(t,s)M^{-1}(s,s)\left(B^{\top}(s)\tilde{P}(s)+S(s,s)\right)\\
\qquad\qquad+Q(t,s)-\left(B^{\top}(s)\tilde{P}(s)+S(s,s)\right)^{\top}M^{-1}(s,s)S(t,s)\\
\qquad\qquad-S^{\top}(t,s)M^{-1}(s,s)\left(B^{\top}(s)\tilde{P}(s)+S(s,s)\right),\forall s\in [t,T].
\end{array}\right.
\end{eqnarray}
Let
\[\tilde{\Upsilon}(s)=M^{-1}(s,s)\left(B^{\top}(s)P(s)+S(s,s)\right), \forall s\in[0,T].\]
Then one can see from (\ref{pc2}) that
\begin{eqnarray} \label{pc3}
&&Q(t,t)-\tilde{\Phi}^{\top}(T,t)\dot{G}(t)\tilde{\Phi}(T,t)-\int_{t}^{T}\tilde{\Phi}^{\top}(s,t)\tilde{Q}_t(t,s)\tilde{\Phi}(s,t) ds\nonumber\\
&=&Q(t,t)-\tilde{\Phi}^{\top}(T,t)\dot{G}(t)\tilde{\Phi}(T,t)-\int_{t}^{T}\tilde{\Phi}^{\top}(s,t)\bigg[\tilde{\Upsilon}^{\top}(s)M_t(t,s)\tilde{\Upsilon}(s)\\
&&+Q_t(t,s)-\tilde{\Upsilon}^{\top}(s)S_t(t,s)-S_t^{\top}(t,s)\tilde{\Upsilon}(s)\bigg]\tilde{\Phi}(s,t) ds,\forall t\in[0,T].\nonumber
	\end{eqnarray}
Using (\ref{2.10}), we have $\tilde{P}(T)=G(T)$ and $\tilde{P}^{\top}(t)=\tilde{P}(t)$ for all $t\in [0,T]$. Moreover, plug  (\ref{pc1}) and (\ref{pc2}), (\ref{pc3}) into  (\ref{2.13}), we then have
\begin{eqnarray*} \label{pc4} 
\left\{\begin{array}{ll}
\dot{\tilde{P}}(t)+A^{\top}(t)\tilde{P}(t)+\tilde{P}(t)A(t)-\left(B^{\top}(t)\tilde{P}(t)+S(t,t)\right)^{\top}M^{-1}(t,t)\left(B^{\top}(t)\tilde{P}(t)+S(t,t)\right)\\
+Q(t,t)-\tilde{\Phi}^{\top}(T,t)\dot{G}(t)\tilde{\Phi}(T,t)-\int_{t}^{T}\tilde{\Phi}^{\top}(s,t)\bigg[Q_t(t,s)+\tilde{\Upsilon}^{\top}(s)M_t(t,s)\tilde{\Upsilon}(s)\\
-\tilde{\Upsilon}^{\top}(s)S_t(t,s)-S_t^{\top}(t,s)\tilde{\Upsilon}(s)\bigg]\tilde{\Phi}(s,t) ds=0, \forall t\in[0,T],\\
\tilde{P}(T)=G(T),
\end{array}\right.
	\end{eqnarray*}
which shows $P\in C([0,T];\mathbb{R}^{n\times n})$ is a symmetric solution of the equilibrium Riccati equation (\ref{1.8}). Then the uniqueness of the linear equilibrium follows from Theorem \ref{peng-theorem5.1}. This completes the proof.
\end{proof}
\section{Concluding remarks}\label{Sec:Conclusion}
The equivalence between LQ problems, two-point boundary value problems and Riccati equations plays a significant role in the study of time-consistent LQ problems. We have extended this equivalence result to the time-inconsistent setting. In contrast to the time-inconsistent Riccati equation obtained from the spike variation, the Riccati equation does not involve the time parameter and thus characterising a single dynamics. Thanks to the unique solvability of the Riccati equation and the equivalence, we have obtained the existence and uniqueness of the linear equilibrium for the time-inconsistent LQ problem.\par 
We have to point out that the uniqueness result is only for linear equilibria and deterministic LQ problems. For general time-inconsistent LQ problems, such as stochastic LQ problems, and general equilibria, the uniqueness result remains open.
\bibliography{SiamFinal}        
\bibliographystyle{ecta}  
\end{document}